%% file: taxes-OR-Journal-v3-ArXiv.tex
\documentclass[opre,nonblindrev]{my-informs3} 
\SingleSpacedXI
\usepackage{natbib}
 \bibpunct[, ]{(}{)}{,}{a}{}{,}%
 \def\bibfont{\small}%
\TheoremsNumberedThrough     %
\ECRepeatTheorems

\EquationsNumberedThrough    %
\input{parts-frontmatter/packages.tex}
\begin{document}
\RUNAUTHOR{Paccagnan and Gairing}

\RUNTITLE{In Congestion Games, Taxes Achieve Optimal Approximation}

\TITLE{In Congestion Games,\\ Taxes Achieve Optimal Approximation}

\ARTICLEAUTHORS{%
\AUTHOR{Dario Paccagnan}
\AFF{Department of Computing, Imperial College London, London, UK,
\EMAIL{d.paccagnan@imperial.ac.uk}} %
\AUTHOR{Martin Gairing}
\AFF{Department of Computer Science, University of Liverpool, Liverpool, UK, \EMAIL{m.gairing@liverpool.ac.uk}}
}

\ABSTRACT{
{
In this work, we consider the problem of minimising the social cost in atomic congestion games. 
For this problem, we provide
tight computational lower bounds along with taxation mechanisms yielding polynomial time algorithms with optimal approximation.}

Perhaps surprisingly, our results show that indirect interventions, in the form of efficiently computed taxation mechanisms, yield the {\it same} performance achievable by the best polynomial time algorithm, even when the latter has full control over the agents' actions.
 It follows that no other tractable approach geared at incentivizing desirable system behavior can improve upon this result, regardless of whether it is based on taxations, coordination mechanisms, information provision, or any other principle. In short: Judiciously chosen taxes achieve optimal approximation. 
Three technical contributions underpin this conclusion.
First, we show that computing the minimum \socialcost is \NP-hard to approximate within a given factor depending solely on the admissible resource costs.
Second, we design a tractable taxation mechanism whose efficiency (price of anarchy) matches this hardness factor, and thus is {worst-case} optimal.
As these results extend to coarse correlated equilibria, any no-regret algorithm inherits the same performances, allowing us to devise polynomial \mbox{time algorithms with optimal approximation.}

}%


\maketitle
\input{parts-body/introductionEC.tex}
\input{parts-body/hardnessapprox.tex}

\input{parts-body/opttaxes.tex}

\input{parts-body/conclusions}
\include{parts-body/ack-biblio}

\newpage
\input{parts-appendix/appendix-EC}

\end{document}

%% file: parts-frontmatter/packages.tex
\usepackage[utf8]{inputenc}
\usepackage{amsmath}
\usepackage{amsfonts}
\usepackage{bbm} 

\usepackage{xcolor}
\definecolor{QuasiBlue}{rgb}{0.03,0.3,0.72}
\colorlet{blue}{QuasiBlue}

\usepackage{pgfplots}
\usepackage{graphicx}
\usepackage{xfrac}
\usepackage{tikz}
\usepackage{hyperref}
\usepackage[capitalise]{cleveref}
\usepackage{xfrac}
\usepackage[font=small, tablename=Table]{caption}
\usepackage{subcaption}
\usepackage{tcolorbox}
\usepackage{tikz-network}
\usepackage{textpos}
\usepackage{xfp}
\usepackage{pgffor}
\usepackage{siunitx}
\usepackage{multirow}
\usepackage{wrapfig}
\usepackage{enumitem}
\usepackage{soul}
\usepackage{bbm}
\usepackage{xspace}
\usepackage{mathrsfs} %
\usepackage{appendix}
\usepackage{quoting}

\newcommand\blfootnote[1]{%
  \begingroup
  \renewcommand\thefootnote{}\footnote{#1}%
  \addtocounter{footnote}{-1}%
  \endgroup
}

\newcommand{\poa}{\mathrm{PoA}}

\newcommand{\systemcost}{social cost\xspace}%
\newcommand{\socialcost}{social cost\xspace}%
\newcommand{\el}{k}
\newcommand{\hfun}{p}
\newcommand{\p}{q}

\newcommand{\llabel}{l}

\renewcommand{\mc}[1]{\mathcal{#1}}
\newcommand{\mb}[1]{\mathbb{#1}}

\newcommand{\PNEcost}{{\tt{NeCost}}}

\newcommand{\mincost}{{\tt{MinCost}}}
\newcommand{\SC}{{\rm{SC}}}

\newcommand{\opt}[1]{{#1}^{\rm opt}}
\newcommand{\NE}[1]{{#1}^{\rm ne}}

\newcommand{\NP}{{\sf NP}}
\newcommand{\Pclass}{{\sf P}}
\pgfplotsset{compat=1.15}

\newlength{\mynegspace}
\newcommand{\mygezero}{>0}
\newcommand{\players}{{agents}}
\newcommand{\player}{{agent}}
\newcommand{\mya}{\alpha}
\newcommand{\myb}{\beta}

\newcommand{\mincon}{\texttt{MinSC}}
\newcommand{\gaplabelcover}{\texttt{GapLabelCover}}

\newcommand{\labelcover}{{\texttt{LabelCover}}}
\newcommand{\lab}{\mc{L}}

\renewcommand{\star}{\ast}

\newcommand{\T}{t}
\newcommand{\lalph}{\alpha}
\newcommand{\ralph}{\beta}
\newcommand{\Tset}{B}

\usetikzlibrary{shapes.misc}
\tikzset{cross/.style={cross out, draw, 
         minimum size=2*(#1-\pgflinewidth), 
         inner sep=0pt, outer sep=0pt}}
         
\usepackage{xr}
\definecolor{pnasblueback}{RGB}{205,217,235}

\definecolor{magenta}{RGB}{0,0,0}

\newcommand*\circled[1]{\tikz[baseline=(char.base)]{\node[shape=circle,draw,inner sep=1pt] (char) {#1};}}

%% file: parts-body/introductionEC.tex
\vspace*{-11mm}
\section{Introduction}
\label{sec:intro}
{
Decision-making in the presence of congestion effects is a central topic in the operations research and game theory literature. %
Within this arena, the celebrated {\it congestion game model} proposed by \cite{rosenthal1973class} offers a fundamental framework to study resource allocation problems prone to congestion, finding fitting applications in transportation, telecommunications, scheduling, and many other disciplines. 
In an atomic congestion game, we are given a finite set of resources and a finite number of \players{}. Every \player{} is endowed with a set of feasible allocations, each corresponding to a subset of the set of resources. These subsets may, for example, correspond to machines where an \player{} wishes to schedule their jobs, or to paths connecting a given origin-destination in a graph. When selecting an allocation, each \player{} incurs an additive cost over the chosen resources, where the cost of a particular resource depends solely on the number of \players{} using it, i.e., on the congestion on that resource. At the system level, the quality of a joint allocation is measured through the \emph{\socialcost} by summing all the \players{} costs.%

Perhaps the most commonly encountered class of congestion games is that of {\it network} congestion games, where each feasible allocation corresponds to a path connecting a given pair of origin-destination nodes over a shared graph. This class of congestion games arises in the study of classical optimization problems in road-traffic, wireless-network, and minimum-power routing -- see \citep{rosenthal1973class,tekin2012atomic,AndrewsAZZ12} respectively.
\blfootnote{$^\dagger$ The authors would like to thank Rahul Chandan, Yiannis Giannakopoulos, Jason Marden, Tim Roughgarden, Maximilian Schiffer and Rahul Savani for valuable comments and suggestions. A preliminary version of this manuscript appeared at the 22nd ACM Conference on Economics \& Computation (EC'21).}
However, congestion games find interesting applications also beyond network-like settings for example in machine scheduling \citep{suri2007selfish}, distributed control \citep{marden2013distributed}, sensor allocation \citep{paccagnan2018distributed}, and factory production \citep{rosenthal1973class}.
\medskip
In this work we are concerned with a {\it system optimal} perspective on congestion games, complementary to much of the existing literature. Specifically, we study the problem of minimising the \socialcost.
We do so in two stages, depending on whether \players{} are modelled as strategic decision-makers or not.

\medskip

We  start by considering the case where \players{} are non-strategic, i.e., the case where we can directly control their allocations for the purpose of minimizing the social cost. This point of view stems from the observation that, in many circumstances, the planner is interested in a globally optimal solution to a congestion game where \players{} can be easily coordinated. Such problems arise commonly, e.g., in load balancing \citep{awerbuch1995load}, power routing \citep{AndrewsAZZ12}, or fleet management in mobility on-demand where a central operator is responsible for routing an entire fleet of cars \citep{zardini2021analysis}. 
However, in spite of the significant scientific interest, the fundamental nature of this model, and its ties to many related questions in operations research, little is known regarding the problem of minimizing the \socialcost in congestion games. %
Most notably, it remains unclear what solution quality  %
can be achieved by efficient (i.e. polynomial time) algorithms, and what are the inherent computational barriers to achieving such goal.
Specifically, no tight hardness of approximation is known to date,  %
nor a polynomial time algorithm with optimal approximation is available.
{\it Our manuscript completely settles these questions in a series of three contributions.} 

\medskip
Our first contribution shows that minimizing the \socialcost in congestion games is \NP-hard to approximate within an explicitly given factor depending solely on the class of resource costs. Prior to our work, no tight computational lower bound was known, even for the special case of linear resource costs. Our second and third contributions provide polynomial time algorithms with an approximation factor matching such hardness barrier. 
Taken together, these results certify that the algorithms we propose are {\it optimal} in the sense that they yield the best possible polynomial time approximation with respect to the {\it worst-case} instance. As such, they leave open the possibility that a different class of algorithms might provide a better performance on an instance-by-instance case.
However, we will see that the approximation we derive is near-optimal for many commonly encountered real-world settings, including the case where resource costs increase not too steeply (e.g., polynomials of low degree, logarithms), or where the uncongested and congested costs are of comparable magnitude (e.g., BPR functions). %
As the proposed approach applies even when agents are strategic, we now present such contributions within this more challenging setting.
\medskip

Decision-making in the presence of strategic \players{} dates back to the seminal works of \cite{morgenstern1953theory} and \cite{nash1950equilibrium}, which laid the foundation for the nascent field of game theory. 
Most notably, their work contributed to establishing the notion of {\it equilibrium}, e.g., Nash equilibrium, which describes the outcome arising from the interaction of strategic decision-makers aiming to minimimize their individual cost functions. Building over this initial work, researchers soon realized that equilibrium allocations often display a significant performance degradation when compared to the corresponding optimal solutions, leading to the definition of the so-called {\it price of anarchy}, commonly used to quantify the worst-case equilibrium efficiency \citep{koutsoupias1999worst}.}
A prototypical example of this issue is provided by road-traffic routing: When drivers choose routes that minimize their own travel time, the aggregate congestion could be much higher compared to that of a centrally-imposed routing. While improved performances could be attained if a coordinator was able to dictate the choices of each \player{}, imposing such control is often infeasible or impossible in these settings, with traffic routing providing just one illustration. Hence, different approaches including coordination mechanisms \citep{ChristodoulouKN09}, Stackelberg strategies {\color{magenta}\citep{BonifaciHS10}}, taxation mechanisms \citep{caragiannis2010taxes}, information provision {\color{magenta}\citep{BhaskarCKS16}}, and alternative methods for sharing resource costs {\color{magenta}\citep{GkatzelisKR16}} have been proposed as  \emph{indirect interventions} to influence the resulting equilibrium efficiency. Amongst these, taxation mechanisms, where \players{} pay an additional cost for utilizing a resource, have attracted significant attention, as witnessed by the growing literature on the topic, see \citep{harks2019pricing} and references therein. 
 {\color{magenta}
Yet, the problem of designing taxation mechanisms optimizing the equilibrium efficiency remains largely open, even for classical settings such as that of atomic congestion games \citep{harks2019pricing}. Additionally, it remains unclear, if, and to what extent, the use of indirect interventions reduces the best achievable performance, thus prompting a natural question:}
{\it Is there any performance degradation incurred when moving from centrally imposed decision-making to the use of \mbox{indirect interventions, such as taxation mechanisms?}} 

\medskip

Perhaps surprisingly, our work answers this question in the negative, showing that \emph{no} performance degradation arises when taxation mechanisms are judiciously designed. %
Specifically, our second contribution derives polynomially-computable taxation mechanisms ensuring that any corresponding equilibrium outcome has an efficiency (price of anarchy) matching the hardness factor derived in this paper, that is, the \emph{same} performance achievable by the best centralized polynomial time algorithm. %
 {\color{magenta} Before this work, optimal taxation mechanisms have been unavailable.}\footnote{
 \label{foot:worst-case}
 Throughout, we remark that the terms ``performance'', ``optimal'' and ``approximation'' are intended with respect to the achievable cost at the worst-case instance, as classical in the analysis of algorithms.}
 {While this result holds for the commonly employed notion of pure Nash equilibrium, it also extends to more general notions of equilibrium that can be computed more efficiently (e.g., correlated and coarse correlated equilibria). Moreover, it extends to no-regret online learning algorithms where \players{} update their allocations and achieve low regret, in the same spirit of the ``total price of anarchy'' pioneered by \cite{BlumHLR08}. Our third and final contribution builds upon this observation to derive polynomial time algorithms achieving the optimal approximation factor. When taken together, the upshot of our work can be summarized as follows:
\medskip
\begin{quoting}[leftmargin=1.35cm]
\emph{In congestion games, judiciously designed taxes achieve optimal approximation, and no other tractable intervention, whether based on coordination mechanisms, information provision, or any other principle can improve upon this result.}	
\end{quoting}}
\medskip
\noindent

\subsection*{\bf Significance and comparison with existing results}
Our work connects with a number of existing results, in addition to closing different open questions%
, as we briefly highlight next. %
{We refer to \cref{subsec:relatedwork} for a more detailed literature review, including a discussion on the continuous flow counterpart of congestion games (nonatomic model).}

The problem of determining computational lower bounds for minimizing the \systemcost in atomic congestion {games} has been initially studied by \cite{MeyersS12}. Since then, a number of works have pursued a similar line of research%
, though no tight bounds were known. 
Our work provides the best possible inapproximability results, completely settling the hardness question for congestion games with general \mbox{underlying resource costs.}

The study of taxation mechanisms has recently received growing attention, especially within the congestion games literature. %
Nevertheless, prior to this work, a methodology to design optimal taxation mechanisms (i.e., mechanisms minimizing the price of anarchy) was not known. %
For congestion games with polynomial costs, \cite{caragiannis2010taxes} (for linear) and \cite{BiloV19} (for polynomial) propose taxes whose efficiency can be quantified a priori. Both papers conjecture that their design is optimal. 
Our work resolves the problem of determining optimal taxes for the broader class of congestion games with non-decreasing and semi-convex resource costs, proving these conjectures as a special case.
\cite{Rough_barrier} studies how lower bounds on the price of anarchy can be derived from computational lower bounds. For congestion games with optimal taxes, we show \mbox{that such an approach does provide \emph{tight} bounds.}
{The best known approximation algorithm for minimizing the \systemcost in congestion games is due to \cite{Makarychev18}, and leverages a natural linear programming relaxation jointly with a randomized rounding scheme. While their result applies to the more general class of optimization problems with a ``diseconomy of scale'', the algorithms we propose here enjoy an equal or strictly better approximation ratio, and can not be further improved, owing to the matching hardness result provided.

\subsection{Congestion game model and taxation mechanisms}
In a congestion game we are given a set of \players{} $\{1,\dots,N\}$, 
and a set of resources $\mc{R}$. Each \player{} can choose a subset of the set of resources which she intends to use. We list all feasible choices for \player{} $i$ in the set $\mc{A}_i\subseteq 2^{\mc{R}}$. The cost for using each resource $r\in\mc{R}$ depends only on the total number of \players{} concurrently selecting that resource, and is denoted with $\ell_r :\mb{N}\rightarrow \mb{R}_{>0}$. Once all \players{} have made a choice $a_i\in\mc{A}_i$, each \player{} incurs a cost obtained by summing the costs of all resources she selected.  Finally, the \systemcost represents the sum of the resource \mbox{costs incurred by all \players{}}
\begin{equation}
\SC(a)=
\sum_{i=1}^N\sum_{r\in a_i} \ell_r(|a|_r),
\label{eq:systemcost}
\end{equation}
where $|a|_r$ denotes the number of \players{} selecting resource $r$ in allocation $a=(a_1,\dots,a_n)$. %
Given an instance $G$ %
{of} a congestion game, we denote with $\mincon$ the problem of globally minimizing the \socialcost in \eqref{eq:systemcost}. 
\medskip

\paragraph{Taxation mechanisms.} 
As self-interested decision making often deteriorates the system performance, taxation mechanisms have been proposed to ameliorate this issue. %
Formally, a taxation mechanism $T:G\times r\rightarrow \tau_r$ associates an instance $G$, and a resource $r\in\mc{R}$ to a taxation function $\tau_r:\mb{N}\rightarrow\mb{R}_{\ge0}$. Note that each taxation function $\tau_r$ is congestion-dependent, that is, it associates the number of \players{} in resource $r$ to the corresponding tax.%
\footnote{
\label{foot:congestion-dep-tolls}
We note that congestion-dependent taxation mechanisms are commonly studied in the literature and have been tested in the real world too, see e.g. \citep{axhausen2021empirical} and references therein. In the context of road traffic routing, increased connectivity and the advent of smart vehicles will further ease their implementation.} 
As a consequence, each \player{} $i$ experiences a cost factoring both the cost associated to the chosen resources, and the \mbox{tax, i.e.,}
\be
\label{eq:agentscost}
C_i(a)=\sum_{r\in a_i} [\ell_r(|a|_r)+\tau_r(|a|_r)].
\ee 
As typical in the literature, we measure the performance of a given taxation mechanism $T$ using the ratio between the \systemcost incurred at the worst-performing outcome, and the minimum \systemcost. Given the self-interested nature of the \players{}, an outcome is commonly described by any of the following classical equilibrium notions: pure or mixed Nash equilibria, {correlated} or coarse correlated equilibria.\footnote{It is worth observing that each class of equilibria appearing in this list is a superset of the previous \citep{roughgarden2015intrinsic}. Therefore, since pure Nash equilibria are guaranteed to exist even in congestion games where taxes are used (they are, in fact, potential games), all mentioned equilibrium's sets are non-empty and thus the notion of \emph{price of anarchy} introduced in \eqref{eq:poadef} \mbox{well defined.}}
When considering pure Nash equilibria, for example, the performance of a taxation mechanism $T$ over the class of congestion games $\mc{G}$ is gauged using the notion of \emph{price of anarchy} originally introduced by \cite{koutsoupias1999worst}, i.e.,  
\begin{equation}
\poa(T) = \sup_{G\in\mc{G}} \frac{\PNEcost(G,T)}{\mincost(G)},
\label{eq:poadef}
\end{equation}
where $\mincost(G)=\min_{a\in\mc{A}} SC(a)$ is the minimum \socialcost for instance $G$, and $\PNEcost(G,T)$ denotes the highest \socialcost at a Nash equilibrium obtained when employing the mechanism $T$ on the game $G$. By definition, $\poa(T)\ge1$ and the lower the price of anarchy, the better performance $T$ guarantees. While it is possible to define the notion of price of anarchy for each and every equilibrium class mentioned, we do not pursue this direction, as we will show that all these metrics coincide within our setting. Thus, we will simply use $\poa(T)$ to refer to the efficiency values of \emph{any} and \emph{all} these equilibrium classes. Finally, we observe that, while taxation mechanisms influence the \players{}' perceived cost, they do not impact the \mbox{expression of the \systemcost, which is still of the form in \eqref{eq:systemcost}.}

\medskip
{
All our results hold for the widely studied case of non-decreasing semi-convex resource costs, corresponding to all resource costs satisfying the following assumption.
 
\begin{assumption}
\label{ass:assumption}
{%
The function $\ell:\mb{N}\rightarrow \mb{R}_{>0}$ is non-decreasing and semi-convex.\footnote{$\ell:\mb{N}\rightarrow \mb{R}_{\mygezero}$ is semi-convex if $x\ell(x)$ is convex, i.e., $(x+1)\ell(x+1)-x\ell(x)\ge x\ell(x)-(x-1)\ell(x-1)$ $\forall x\in\mb{N}, x\ge2$.}}
\end{assumption}

\paragraph{Notation.}
\label{par:notation}
In our work $\mathbb{N}$, $\mathbb{N}_0$, $\mathbb{R}$, $\mathbb{R}_{\ge0}$,  $\mathbb{R}_{>0}$ denote the sets of natural numbers, natural numbers including zero, real numbers, non-negative real numbers, and positive real numbers. Further, $\text{Poi}(x)$ denotes a Poisson distribution with parameter $x\in\mb{R}_{>0}$.
{In the remainder of the manuscript, we extend the domain of definition of all resource costs from $\mb{N}$ to $\mb{N}_0$ by setting their value to zero. This is without loss of generality. Indeed, the value of $\ell_r(0)$, does not play any role with respect to all quantities we have introduced thus far, e.g., the \systemcost in \eqref{eq:systemcost}, %
since $\ell_r(|a|_r)$ is always evaluated for $|a|_r\ge 1$. %
However, doing so will allow us to simplify notation.}\footnote{\label{foot:simplify_b(0)}
For example, it will allow us to write the numerator of \eqref{eq:approxfact} compactly as $\mb{E}_{P\sim\text{{Poi}}(x)} [P\ell(P)]=\sum_{k=0}^\infty k\ell(k) {x^k e^{-x}}/{k!}$, as opposed to $\sum_{k=1}^\infty k\ell(k) {x^k e^{-x}}/{k!}$ when $\ell(0)$ is undefined. Note here the different extremes of summation.} }

\subsection{Our contributions}
\label{subsec:contributions}
The resounding message contained in this work can be summarized as follows:
In congestion games, optimally designed taxation mechanisms can be tractably computed, and achieve the same performance of the best centralized polynomial time algorithm. {Further, such approximation is near-optimal in a variety of commonly-encountered settings.} Three technical contributions substantiate this claim, and are further discussed in the ensuing paragraphs:%
\medskip
\begin{itemize}[leftmargin=5.5mm]
{\it
\item[ i)] We prove a tight \NP-hardness result for minimizing the \socialcost.
\item[ ii)] We design a tractable taxation mechanism whose price of \mbox{anarchy matches the hardness factor.}
\item[ iii)] 
{We obtain a polynomial time algorithm with the best possible approximation ratio for the social cost. We do so combining the previous results with existing algorithms (e.g., no-regret dynamics).}}
\end{itemize}
\medskip
\noindent
{Results {\it ii)} and {\it iii)} extend to network congestion games, as we discuss in the conclusions.}
\medskip
\noindent{\bf Inapproximability of minimum \socialcost.} Our first contribution is concerned with determining tight inapproximability results for the problem of minimizing \socialcost in congestion games. %
Our hardness result applies already in the setup where all resources feature the same cost.%
\medskip
\begin{theorem}
\label{thm:hardness}
{Let $\ell:\mb{N}\rightarrow\mb{R}_{\mygezero}$ satisfing Assumption~\ref{ass:assumption} be given.}
In congestion games where all resources feature the same cost $\ell$,
\mincon~is \NP-hard to approximate within any factor smaller than 
\be
\rho_\ell=\sup_{x\in\mb{N}}\frac{\mb{E}_{P\sim\text{\emph{Poi}}(x)} [P\ell(P)]}{x\ell(x)},
\label{eq:approxfact}
\ee
{where we define $\ell(0)=0$.}
{If $\rho_\ell=\infty$, then \mincon~is \NP-hard to approximate within any finite factor.}
\end{theorem}
\medskip
\noindent
{Naturally, \cref{thm:hardness} applies directly to richer classes of congestion games, whereby resource costs can differ. {For example, if resource costs belong to a set of functions $\mathscr{L}$, 
 then \mincon~is \NP-hard to approximate within any factor smaller than that produced by the worst function in $\mathscr{L}$, i.e., $\sup_{\ell\in\mathscr{L}}\rho_{\ell}$. } 
 As a special case, we obtain hardness results for the thoroughly studied class of polynomial congestion games with maximum degree $d$. %
 In this case, the highest degree monomial $x^d$ determines the worst factor (see \cref{subsec:bell-number}), which reduces to the $(d+1)$'st Bell number, as summarized in the following statement. 
\medskip
 \begin{corollary}
 \label{cor:polyCG}
 In congestion games with resource costs obtained by non-negative combinations~of $1$, $\dots$, $x^d$, \mbox{\mincon~is \NP-hard to approximate within a factor smaller than the $\!(d\!+\!1)\!$'st Bell number.}
 \end{corollary}
 
 {%
 We note that Bell numbers grow as function of the degree $d$, so that the corresponding inapproximability factor also grows with $d$. %
 However, for many problems of interest, resource costs feature further properties that can be exploited to reduce such factor. This is interesting, since the algorithms we propose will provide approximations matching this factor. For example, in road-traffic routing, the commonly employed BPR functions take the form $\ell_r(x)=\mya_rx^4+\myb_r$ where $\mya_r\ll \myb_r$ for all resources, see for example the commonly employed Sioux Fall network and other instances from \citep{TransNet}. In this context, a direct application of \cref{thm:hardness} provides inapproximability factors significantly closer to one. We elucidate this point in the following corollary where we look at congestion games with affine resource costs, as this allows to present a concise analytical result. However, a similar result holds for BPR functions, as demonstrated in Figure \ref{fig:inapprox} (left).

\begin{corollary}
\label{cor:affine_CG_refined}
	In congestion games with affine resource costs of the form $\ell_r(x)=\mya_r x+ \myb_r$, where $\max_r {\mya_r}/({\mya_r+\myb_r})\le q$, \mincon~is \NP-hard to approximate within a factor smaller than $1+ q$.
\end{corollary}

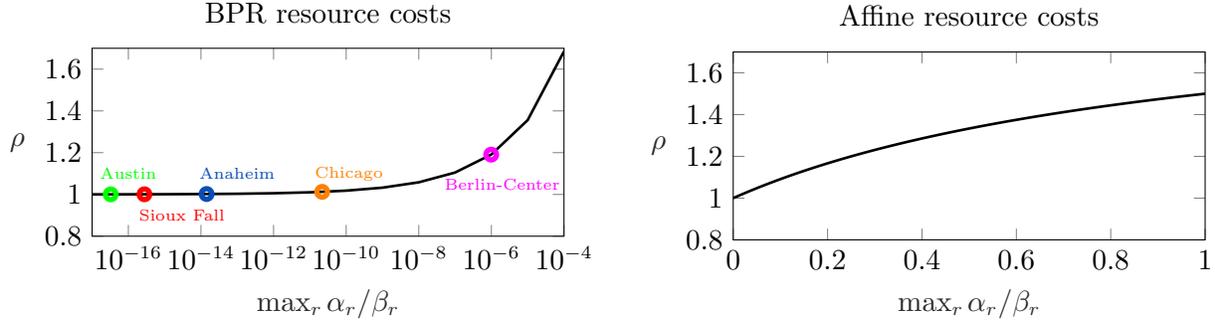
\begin{figure}[t!]
	\newlength\figureheight 
	\newlength\figurewidth 
	\setlength\figureheight{2.5cm} 
	\setlength\figurewidth{.38\linewidth} 
	\input{inapprox_BPR.tikz}
	\quad
	\input{inapprox_lin.tikz}
	\caption{Inapproximability factors for congestion games with BPR resource costs $\ell_r(x)=\mya_r x^4 +\myb_r$ (left) and affine resource costs $\ell_r(x)=\mya_r x +\myb_r$ (right), parametrized by $\max_r \mya_r/\myb_r$. Commonly encountered instances, e.g., those taken from \citep{TransNet} feature $\mya_r\ll \myb_r$ yielding an inapproximability factor close to one, see markers on the left panel.}
	\label{fig:inapprox}
\end{figure}

The latter result shows that the inapproximability factor depends only on the ratios between the coefficients $\mya_r$ and $\myb_r$. To see this, denote $\max_r {\mya_r}/{\myb_r}= \kappa$, and observe that $\max_r {\mya_r}/({\mya_r+\myb_r})= {\kappa}/({\kappa+1})$ when $\myb_r\neq0$. Consequently, the  inapproximability factor $1+q$ can be rewritten more conveniently as $1+{\kappa}/({\kappa+1})$, see \cref{fig:inapprox} (right) for a plot.
In the regime where $\mya_r\ll \myb_r$ for all resources, we have $\kappa\ll 1$, so that the inapproximability factor is close to one. When this ratio grows, the corresponding factor also grows up until it reaches its maximum when $\myb_r\rightarrow 0$ for some resource. In other words, the inapproximability factor depends on the ratio between the speed at which congestion translates into cost (here $\mya_r$) and the uncongested cost (here $\myb_r$). Interestingly, in many setting the ratio ${\mya_r}/{\myb_r}\ll 1$ , thus ensuring that the forthcoming algorithms will provide near-optimal approximation.
}
 
 \medskip

{\color{magenta}
Regarding our techniques, the computational lower bounds are shown by reducing the problem of minimizing the \systemcost in congestion games from the \emph{gap label cover} problem implicitly defined in \cite{Feige98}. A central tool we leverage in our reduction is a gadget called \emph{partitioning system}, generalizing that in \cite{Feige98}. The challenge in defining this object stems from the fact that, in congestion games, the cost experienced over each resource depends on the actual number of \players{} selecting that resource, contrary to \cite{Feige98} where the welfare accrued depends solely on whether that resource is selected or not.
The expression of the hardness factor \eqref{eq:approxfact} falls from the very definition of such object, which must be duplicated multiple times and carefully arranged in our construction to ensure hardness of the underlying instance. {
Conceptually, partitioning systems are constructed so that, in our reduction, \players{} have two types of allocations: either they coordinate themselves on an allocation with low \socialcost, or they opt for an allocation with high social cost. The inapproximability factor in \eqref{eq:approxfact} is precisely the ratio between these ``high'' and ``low'' social costs, see \cref{subsubsec:partsys}.
} 
Partitioning systems have been recently applied also by \cite{DudyczMMS20} and \cite{barman2021tight} in the context of approval voting, and for a generalization of the maximum coverage problem. To the best of our knowledge, our work is the first that employs such gadgets to prove hardness of cost minimization problems. 
}

\medskip
\noindent{\bf Taxes achieve optimal approximation. }
Our second contribution provides a technique to efficiently design taxation mechanisms {\color{magenta}whose efficiency (price of anarchy)} matches the hardness result. The following statement is a succinct variant of that in \cref{thm:opttaxes} where we also give an expression for the taxation mechanism.%
\begin{theorem}
\label{thm:taxes_poa}
Consider the class of congestion games where each resource cost belongs to a set of functions $\mathscr{L}$ {satisfying Assumption~\ref{ass:assumption}}. Then, for any given $\varepsilon>0$, it is possible to efficiently compute a taxation mechanism whose price of anarchy is upper bounded by $\sup_{\ell\in\mathscr{L}} \rho_{\ell}+\varepsilon$. The result holds for pure/mixed Nash and \mbox{correlated/coarse correlated equilibria.}
\end{theorem}
\medskip
\noindent
The extension to coarse correlated equilibria is significant as it gives performance bounds that apply not only to Nash equilibria, but also whenever \players{} revise their action and achieve low regret. This imposes much weaker assumptions on both the game and its participants' behaviour \citep{roughgarden2015intrinsic}.
Whilst the theorem applies to a broad class of resource costs, when specialized to polynomial congestion games of degree $d$, it allows to efficiently design taxation mechanisms whose price of anarchy equals the $(d+1)$'st Bell number plus $\varepsilon$, for any \mbox{arbitrarily small choice of $\varepsilon>0$.} 

{\color{magenta} 
We derive taxes whose efficiency matches the hardness factor leveraging two chief ingredients: a parametrized class of taxation mechanisms satisfying a key recursion, and a suitably defined convex optimization program. The convex optimization problem we consider corresponds to a modification of the original \mincon, whereby we relax the integrality constraints and replace the cost $x\ell_r(x)$ produced by each resource with $\mb{E}_{P\sim\text{Poi}(x)}[P\ell_r(P)]$, see \eqref{eq:cvxprogram}. The solution vector of this program (which is provably convex) is used as set of parameters for the class of mechanisms previously defined. The performance bound on the price of anarchy is finally shown through a smoothness-like approach, where we leverage both the expression of the mechanisms and the solution vector of the convex program. 
}

\medskip

\noindent{\bf Tight polynomial algorithms.} Since the result in \cref{thm:taxes_poa} holds also for correlated/coarse correlated equilibria, we can leverage existing polynomial time algorithms to compute such equilibria, and inherit an approximation ratio matching the corresponding price of anarchy.
The following statement summarizes this result, while the ensuing discussion provides two possible approaches to do so. %
We remark that \cref{cor:polyalgo} is a direct consequence of \cref{thm:taxes_poa} and polynomial computability of correlated equilibria \citep{JiangL15}. The fact that the resulting approximation factor we obtain always matches or strictly improves upon that of \cite{Makarychev18} is shown in \cref{subsec:comparison-sviri}.
\medskip
\begin{corollary}
\label{cor:polyalgo} 
	Consider congestion games where each resource cost belongs to a set of functions $\mathscr{L}$ {satisfying Assumption~\ref{ass:assumption}}.
	Then, for any $\varepsilon>0$, there exists a polynomial time algorithm computing an allocation with cost lower than \mbox{$(\sup_{\ell\in\mathscr{L}} \rho_{\ell} +\varepsilon)\cdot \min_{a\in\mc{A}}\SC(a)$.}
\end{corollary}
\medskip

The approximation ratio presented in \cref{cor:polyalgo} can be achieved, for example, as follows. Given a desired tolerance $\varepsilon>0$, we design a taxation mechanisms ensuring a price of anarchy of $\sup_{\ell\in\mathscr{L}} \rho_{\ell} +\varepsilon/2$, which can be done in polynomial time thanks to \cref{thm:taxes_poa}. We use such taxation mechanisms, and compute an exact correlated equilibrium in polynomial time leveraging the result of \cite{JiangL15}, who propose a variation of the Ellipsoid Against Hope algorithm from \cite{papadimitriou2008computing}. Remarkably, \cite{JiangL15} guarantees that the resulting correlated equilibrium has polynomial-size support, i.e., it places non-zero probability only over a polynomial number of pure strategy profiles. Hence, we compute a correlated equilibrium and enumerate all pure strategy profiles in its support, identifying with $a^*$ that with lowest cost.
Since the price of anarchy bounds of \cref{thm:taxes_poa} hold for correlated equilibria, the pure strategy profile $a^*$ inherits a matching \mbox{(or better) approximation ratio.}

Alternatively, one can employ the same taxation mechanism as in the above, and let \players{} simultaneously revise their action for $\T$ rounds employing a no-regret algorithm. Well-known families of such algorithms include {\it Multiplicative-Weights}, and {\it Follow the Perturbed}/{\it Regularized Leader} whereby the average regret decays to an arbitrary $\varepsilon$ in a polynomial %
 number of rounds, see \cite{cesa2006prediction,KalaiV05}. 
We keep track of the pure strategy profile $a^*_{\T}$ with the lowest \systemcost encountered during the $\T$ rounds of any such algorithm t. Owing to \cite[Thm 3.3]{roughgarden2015intrinsic}, its approximation ratio is upper bounded by the corresponding price of anarchy plus an error term that goes to zero with the same rate of the average regret. Waiting for polynomially \mbox{many rounds suffices to reduce the error as desired.}

\subsection{Related work}
\label{subsec:relatedwork}
As our work provides tight computational lower bounds, optimal taxation mechanisms, and polynomial time algorithms with the best possible approximation, we review the relevant literature connected with these three areas in the ensuing paragraphs. We conclude comparing our results with those available for the continuous flow counterpart of congestion games (nonatomic model).

\medskip
\paragraph{\bf Computational lower bounds.}
The study of computational lower bounds for minimizing the \systemcost in congestion games has been pioneered by  \cite{MeyersS12}, though remarkable precursors include \cite{ChakrabartyMN05}, as well as \cite{BlumrosenD06} who considered a notably different model whereby \player{}-specific cost functions are utilized. %
Relative to the classical model of congestion games with convex non-decreasing resource costs, Meyers and Schulz show that minimizing the \systemcost is strongly \NP-hard, and it is hard to approximate within any finite factor, unless \Pclass=\NP. Identical results are shown for network congestion games. Whilst this might feel as a contradiction of our results, it is worth noting that their analysis allows for resource costs to be adversarily selected amongst any convex non-decreasing function%
. On the contrary, our result can be thought of as a refined version of theirs, whereby the computational lower bound we derive is parametrized by the class of admissible resource costs. Naturally, we recover their inapproximability result when resource costs can be arbitrarily selected amongst convex non-decreasing functions.\footnote{For example $\mincon$ is \NP-hard to approximate within any finite factor when $b(x)=e^x$. This is an immediate consequence of \cref{thm:hardness} and the fact that $\rho_b=\infty$ for exponentially increasing resource costs.}%

Motivated by the possibility to translate computational lower bounds to lower bounds on the price of anarchy, \cite{Rough_barrier} also studied this problem. %
Relative to polynomial resource costs of maximum degree $d$ and non-negative coefficients, he showed that minimizing the \systemcost is inapproximable within any factor smaller than $(\beta d)^{d/2}$, for some constant $\beta>0$. 
In this setting, even without taxes, equilibria with much better {performance} are guaranteed to exist. In particular, the \emph{price of stability} (measuring the quality of the best-performing equilibrium) is known to grow only linearly with the degree $d$ \citep{CG16}.  While coordinating the \players{} to one such \emph{good} equilibrium is highly desirable, our hardness result implies that this cannot be achieved in polynomial time.

Spurred by the inapproximability results of \cite{MeyersS12}, and \cite{Rough_barrier}, a number of works have focused on restricting the allowable class of problems: \cite{PiaFM17} study totally unimodular congestion games and show \NP-hardness for the asymmetric case;  \cite{Castiglioni2020signaling} show \NP-hardness even in singleton congestion games with affine resource costs. Similar questions have been explored for online versions of the problem by \cite{KlimmST19}. %
\medskip

\paragraph{\bf Taxation mechanisms.}
Different approaches, such as coordination mechanisms \citep{ChristodoulouKN09}, Stackelberg strategies \citep{Fotakis10a}, taxation mechanisms \citep{caragiannis2010taxes}, signalling \citep{BhaskarCKS16}, cost-sharing strategies \citep{GairingKK20}%
, and many more, have been proposed to cope with the performance degradation associated to selfish decision making.
Amongst them, taxation mechanism have attracted significant attention thanks to their ability to indirectly influence the resulting system performance.  While the study of taxation mechanisms in road-traffic networks was initiated by \cite{pigou}, who utilized a continuous flow model (nonatomic congestion games), the design of taxation mechanisms in (atomic) congestion games was pioneered much more recently by \cite{caragiannis2010taxes}. In this respect, \cite{caragiannis2010taxes} and many of the subsequent works, build on the solid theoretical ground developed in the years subsequent to the definition of the price of anarchy \citep{koutsoupias1999worst}, including exact knowledge of the price of anarchy in congestion games with linear \citep{christodoulou2005price,awerbuch2005price} and polynomial \citep{aland2006exact} resource costs, the advent of the smoothness framework \citep{roughgarden2015intrinsic}, as well as primal dual approaches \citep{bilo2018unifying,chandan19optimaljournal}. 

While these results provide us with a strong theory to quantify the price of anarchy, prior to this work, the design of optimal taxation {mechanisms} (i.e., taxation mechanisms minimizing the price of anarchy) has been an open question even when restricting attention to linear resource costs.
Most notably, \cite{caragiannis2010taxes} considers linear congestion games and designs taxation mechanisms achieving a price of anarchy of $2$ for mixed Nash equilibria. More recently, \cite{BiloV19} extend their results to polynomial congestion games of degree $d$ achieving a price of anarchy (for coarse correlated equilibria) equal to the $(d+1)$'st Bell number. %
Our work resolves the problem of designing optimal taxation mechanisms for more general congestion games with semi-convex non-decreasing resource costs, and, as a special case, 
 shows that the mechanisms proposed for linear \citep{caragiannis2010taxes} and polynomial \citep{BiloV19} resource costs are optimal, as conjectured by the authors.

Perhaps closest in spirit to our work, is the recent result by \cite{paccagnan2019incentivizing}, whereby the authors leverage a tractable linear programming formulation to design optimal taxation mechanisms that utilize solely \emph{local} information. Naturally, the corresponding values of the optimal price of anarchy they achieve are inferior to ours (here, we design optimal taxation mechanisms without any restrictions on what type of information we use), though the efficiency values derived therein are remarkably close to the optimal values obtained here. For example, for affine (resp. quadratic) congestion games, they achieve an optimal price of anarchy of 2.012 (resp. 5.101), to be compared with a value of 2 (resp. 5). This suggests that restricting the attention to taxation mechanisms utilizing solely local information is sufficient to match almost exactly the performance of %
the best polynomial time algorithm.

We conclude observing that similar questions have been considered for variants of the classical setup studied here. For example,
\cite{fotakis2008cost} focuses on symmetric network congestion games, \cite{HarksKKM15, HoeferOS08, JelinekKS14} focus on taxing a subset of the resources. 
Finally, we remark that optimal taxation mechanisms can be easily derived for \emph{non-atomic} congestion games, where \players{} have only an infinitesimal impact on the congestion. 
In this setting, it is known that marginal cost taxes incentivize optimal behaviour \citep{pigou}. On the contrary, in the atomic regime the same marginal cost taxes do not improve \mbox{- and instead significantly deteriorate - the resulting system efficiency \citep{paccagnan2019incentivizing}.}
\medskip
\paragraph{\bf Approximation algorithms.}
A number of polynomial time algorithms have been proposed for approximating the minimum \socialcost in congestion games and their network counterpart as discussed in \cite{Makarychev18, AndrewsAZZ12, HarksOV16} and references therein.
The best known approximation is due to \cite{Makarychev18} who use randomization to round the solution of a natural linear programming relaxation. They provide a general expression for the resulting approximation factor as a function of the allowable resource costs.
Related works have also considered modifications of the classical setup: \cite{HarksOV16} provide approximation algorithms for polymatroid congestion games, whereas \cite{KumarMPS09} considers scheduling problems on unrelated machines. While the result of \cite{Makarychev18} holds for the more general class of optimization problems with ``diseconomy of scale'', the approximation ratios we obtain here always match or strictly improve upon theirs.
\medskip
{
\label{page:nonatomic-comparison}
\paragraph{\bf Comparison with nonatomic congestion games.}
The continuous flow counterpart of congestion games, commonly referred to as the {\it nonatomic} model, and its corresponding Wardrop equilibrium were proposed by \cite{wardrop1952road} for his studies on road-traffic routing. %
 The two models are identical in spirit, except that the nonatomic variant postulates the presence of infinitely many \players{}, each controlling an infinitesimal portion of the traffic flow. Albeit apparently minor, this modification has significant impact both on the applicability %
 of the model and on the corresponding analysis. 
 
 Regarding the model's applicability, %
 nonatomic congestion games require each \player{}'s decision to have negligible impact on the congestion levels. While this simplifying assumption is acceptable in many settings, e.g., road-traffic routing, there are many other congestion-prone systems for which this is simply not the case, e.g., scheduling, power routing, sensor allocation and many more.\footnote{Already within the road-traffic realm, this assumption needs to be carefully considered. For example, in settings where roads have small capacity (e.g., in a city centre) few vehicles can impact the congestion levels significantly.} %
 
Regarding the analysis, nonatomic congestion games provide a platform that is, in general, easier to investigate. %
Specifically, the problem of minimizing the \socialcost can be formulated as a convex optimization program (under standard semi-convexity assumptions on the resource costs), and thus efficiently solved to optimality even for very large networks. A commonly employed approach is based on the Frank-Wolfe algorithm which leverages efficient algorithms for shortest-path calculations \citep{geisberger2012exact}. Contrary to that, in the atomic case, minimizing the social cost is $\NP$-hard and remains $\NP$-hard to approximate within a factor that we give  in \eqref{eq:approxfact}.
A similar conclusion holds when turning attention to the design of optimal taxation mechanisms, owing to the uniqueness of the equilibrium flows in the nonatomic model. For this setting, there are well-known taxation mechanisms, referred to as marginal cost taxes, that can be efficiently-computed and incentivize optimal behaviour \citep{pigou}. The question is significantly more challenging in the atomic case, where a mechanism is confronted with a multiplicity of equilibria. 
{Here, the same marginal cost taxes do not improve - and instead can deteriorate - the resulting system efficiency \citep{paccagnan2019incentivizing}.}\footnote{
\label{foot:marginalcosttolls}
 While this statement might feel as a contradiction of the results on the convergence of Nash equilibria to Wardrop equilibria for large number of  \players{} (see following paragraph), we stress that, for such convergence results to hold, each \player{} needs to have infinitesimal impact on the congestion in the limit.  However, there are problems settings in which this assumption is not satisfied. In these cases, i.e., when the number of \players{} grow but their impact on the congestion levels does not vanish, \citep{paccagnan2019incentivizing} shows that employing the marginal cost mechanisms yields a worse price of anarchy than that obtained without levying {\it any} toll.
}
Nonetheless, our results show that efficiently-computable taxation mechanisms can still provide an approximation equal to that obtained by the best centralized polynomial time algorithm. Further, they show that such approximation is near optimal for interesting classes of problems where the nonatomic model is typically utilized (see \Cref{fig:inapprox} and related discussion).

 Finally, we highlight an important connection between non-atomic congestion games and the limiting behaviour of atomic congestion games with large number of \players{}. Specifically, \cite{haurie1985relationship} first, \cite{paccagnan2018nash}
 and \cite{cominetti2020approximation} more recently, show that the set of Nash equilibria of an atomic congestion game converges to the set of Wardrop equilibria of its nonatomic counterpart if the number of \players{} grow large and each of them has diminishing impact on the congestion. It is important to highlight, however, that estimates on the distance between Nash and Wardrop equilibria are available only under further restrictive assumptions on the resource costs, and such estimates depend heavily on the number of \players{}, the parameters and the structure of the game. As a result, convergence may require a very large number of \players{}, so that the nonatomic model should be carefully used as a surrogate for its atomic counterpart.\footnote{For example, \cite[Thm 6.2]{cominetti2020approximation} shows that, strict increasingness of the resource costs is necessary to provide such bounds. Under this (and other milder) assumptions they show that, the distance between Nash and Wardrop equilibria decreases as $\theta \cdot n^{-1/2}$, when \players{} influence decreases as $1/n$. Here, $\theta$ crucially depends on the Lipschitz constant of the resource costs, the total traffic flow, and the size of the strategy space.} 
 Similar studies have been carried out also for stochastic counterparts of atomic congestion games where a large number of \players{} participate in the game, each with small probability \citep{cominetti2020approximation}. Therein the authors show that the congestion levels at a Nash equilibrium converge to a Poisson distributed random variable with mean equal the flow at a Wardrop equilibrium of a modified game. Interestingly, this model allows for uncertainty in the traffic equilibrium with a resulting Poisson distribution closely describing the real-world congestion levels, as demonstrated by the same authors. When applied to this setting, our results are valuable in that they continue to hold {\it robustly} against different possible realizations of the traffic demand, and not only for the expected traffic conditions described by the Wardrop equilibrium.
}}

%% file: inapprox_BPR.tikz
\definecolor{mycolor1}{rgb}{1.00000,0.00000,1.00000}%
\definecolor{mycolor2}{rgb}{1.00000,0.500000,0.00000}%
\begin{tikzpicture}

\begin{axis}[%
width=\figurewidth,
height=\figureheight,
title={BPR resource costs},
at={(1.011in,0.642in)},
scale only axis,
xmode=log,
xmin=1e-17,
xmax=0.0001,
xminorticks=true,
xlabel style={font=\color{white!15!black}},
xlabel={$\max_r \mya_r/\myb_r$},
ymin=0.8,
ymax=1.7,
ylabel style={font=\color{white!15!black},rotate=-90},
ylabel={$\rho$},
axis background/.style={fill=white}
]

\node at (1e-16,1.1) {\tiny \color{green} Austin};
\node at (3e-15,0.9) {\tiny \color{red} Sioux Fall};
\node at (1e-13,1.1) {\tiny \color{blue} Anaheim};

\node at (2e-6,1.05) {\tiny \color{mycolor1} Berlin-Center};
\node at (1.2e-10,1.1) {\tiny \color{mycolor2} Chicago};

\addplot [color=black, line width=1.0pt, forget plot]
  table[row sep=crcr]{%
1e-18	1.00018022168164\\
1e-17	1.00032049948883\\
1e-16	1.00056998502573\\
1e-15	1.00101374248973\\
1e-14	1.00180319135654\\
1e-13	1.00320807709397\\
1e-12	1.00570960103355\\
1e-11	1.01016827908719\\
1e-10	1.01812958400838\\
1e-09	1.03239030123785\\
1e-08	1.05807834116144\\
1e-07	1.10480925415307\\
1e-06	1.19129553571172\\
1e-05	1.35604634570781\\
0.0001	1.68522254034357\\
};
\addplot [color=red, line width=2.0pt, only marks, mark=o, mark options={solid, red}, forget plot]
  table[row sep=crcr]{%
2.76999999920885e-16	1.00073537238728\\
};
\node[right, align=left, font=\color{red}]
at (axis cs:0,1.031) {Sioux Falls};
\addplot [color=blue, line width=2.0pt, only marks, mark=o, mark options={solid, blue}, forget plot]
  table[row sep=crcr]{%
1.42889803383631e-14	1.00197158936726\\
};
\node[right, align=left, font=\color{blue}]
at (axis cs:0,1.032) {Anaheim};
\addplot [color=green, line width=2.0pt, only marks, mark=o, mark options={solid, green}, forget plot]
  table[row sep=crcr]{%
3.26662634436557e-17	1.00043089175349\\
};
\node[right, align=left, font=\color{green}]
at (axis cs:0,1.03) {Austin};
\addplot [color=mycolor1, line width=2.0pt, only marks, mark=o, mark options={solid, mycolor1}, forget plot]
  table[row sep=crcr]{%
9.8808847473199e-07	1.19069464953262\\
};
\node[right, align=left, font=\color{mycolor1}]
at (axis cs:0,1.211) {Berlin};
\addplot [color=mycolor2, line width=2.0pt, only marks, mark=o, mark options={solid, mycolor2}, forget plot]
  table[row sep=crcr]{%
2.18032536901292e-11	1.01236497363008\\
};
\node[right, align=left, font=\color{mycolor2}]
at (axis cs:0,1.042) {Chicago};
\end{axis}

%\begin{axis}[%
%width=7.778in,
%height=5.833in,
%at={(0in,0in)},
%scale only axis,
%xmin=0,
%xmax=1,
%ymin=0,
%ymax=1,
%axis line style={draw=none},
%ticks=none,
%axis x line*=bottom,
%axis y line*=left
%]
%\end{axis}
\end{tikzpicture}%

%% file: inapprox_lin.tikz
\begin{tikzpicture}

\begin{axis}[%
width=\figurewidth,
height=\figureheight,
title={Affine resource costs},
at={(1.011in,0.642in)},
scale only axis,
xmin=0,
xmax=1,
xlabel style={font=\color{white!15!black}},
xlabel={$\max_r \mya_r/\myb_r$},
ymin=0.8,
ymax=1.7,
ylabel style={font=\color{white!15!black},rotate={-90}},
ylabel={$\rho$},
axis background/.style={fill=white}
]
\addplot [color=black, line width=1.0pt, forget plot]
  table[row sep=crcr]{%
0	1\\
0.01	1.00990099009901\\
0.02	1.01960784313725\\
0.03	1.02912621359223\\
0.04	1.03846153846154\\
0.05	1.04761904761905\\
0.06	1.05660377358491\\
0.07	1.06542056074766\\
0.08	1.07407407407407\\
0.09	1.08256880733945\\
0.1	1.09090909090909\\
0.11	1.0990990990991\\
0.12	1.10714285714286\\
0.13	1.11504424778761\\
0.14	1.12280701754386\\
0.15	1.1304347826087\\
0.16	1.13793103448276\\
0.17	1.14529914529915\\
0.18	1.15254237288136\\
0.19	1.15966386554622\\
0.2	1.16666666666667\\
0.21	1.17355371900826\\
0.22	1.18032786885246\\
0.23	1.1869918699187\\
0.24	1.19354838709677\\
0.25	1.2\\
0.26	1.20634920634921\\
0.27	1.21259842519685\\
0.28	1.21875\\
0.29	1.22480620155039\\
0.3	1.23076923076923\\
0.31	1.23664122137405\\
0.32	1.24242424242424\\
0.33	1.24812030075188\\
0.34	1.25373134328358\\
0.35	1.25925925925926\\
0.36	1.26470588235294\\
0.37	1.27007299270073\\
0.38	1.27536231884058\\
0.39	1.28057553956835\\
0.4	1.28571428571429\\
0.41	1.29078014184397\\
0.42	1.29577464788732\\
0.43	1.3006993006993\\
0.44	1.30555555555556\\
0.45	1.31034482758621\\
0.46	1.31506849315068\\
0.47	1.31972789115646\\
0.48	1.32432432432432\\
0.49	1.32885906040268\\
0.5	1.33333333333333\\
0.51	1.33774834437086\\
0.52	1.34210526315789\\
0.53	1.34640522875817\\
0.54	1.35064935064935\\
0.55	1.35483870967742\\
0.56	1.35897435897436\\
0.57	1.36305732484076\\
0.58	1.36708860759494\\
0.59	1.37106918238994\\
0.6	1.375\\
0.61	1.37888198757764\\
0.62	1.38271604938272\\
0.63	1.38650306748466\\
0.64	1.39024390243902\\
0.65	1.39393939393939\\
0.66	1.39759036144578\\
0.67	1.40119760479042\\
0.68	1.4047619047619\\
0.69	1.40828402366864\\
0.7	1.41176470588235\\
0.71	1.41520467836257\\
0.72	1.41860465116279\\
0.73	1.42196531791908\\
0.74	1.42528735632184\\
0.75	1.42857142857143\\
0.76	1.43181818181818\\
0.77	1.43502824858757\\
0.78	1.43820224719101\\
0.79	1.4413407821229\\
0.8	1.44444444444444\\
0.81	1.4475138121547\\
0.82	1.45054945054945\\
0.83	1.45355191256831\\
0.84	1.45652173913043\\
0.85	1.45945945945946\\
0.86	1.46236559139785\\
0.87	1.46524064171123\\
0.88	1.46808510638298\\
0.89	1.47089947089947\\
0.9	1.47368421052632\\
0.91	1.47643979057592\\
0.92	1.47916666666667\\
0.93	1.48186528497409\\
0.94	1.48453608247423\\
0.95	1.48717948717949\\
0.96	1.48979591836735\\
0.97	1.49238578680203\\
0.98	1.4949494949495\\
0.99	1.49748743718593\\
1	1.5\\
};
\end{axis}
\end{tikzpicture}%

%% file: parts-body/hardnessapprox.tex
\subsection{Roadmap}
The remainder of the manuscript is organised as follows. In \Cref{sec:hardness} we prove \Cref{thm:hardness} (hardness of approximation), and \Cref{cor:polyCG,cor:affine_CG_refined} (hardness factor for polynomial resource costs). In \Cref{sec:taxes} we state and prove \Cref{thm:opttaxes}, an enriched version of \Cref{thm:taxes_poa} whereby the expression of the taxation mechanism is also given. A discussion on future research directions concludes the manuscript.

\section{\NP-hardness of approximation}
\label{sec:hardness}
In this section we prove \cref{thm:hardness}, i.e., we show that approximating the minimum \systemcost below the factor $\rho_\ell$ defined in \eqref{eq:approxfact} is \NP-hard already for congestion games where all resource costs are identical to $\ell$. The proof is based on a reduction from \gaplabelcover, where we make use of a generalization of Feige's partitioning system \citep{Feige98}. 
We proceed as follows: In \cref{subsec:ingredients} we introduce \gaplabelcover~and, independently, the partitioning system. In \cref{subsec:reduction} we present the reduction and in \cref{subsec:actualproof} prove \cref{thm:hardness}. \cref{subsec:bell-number} shows that $\sup_{\ell\in\mathscr{L}} \rho_\ell$ reduces to the $(d+1)$'st Bell number when resource costs are polynomials of maximum degree $d$, as claimed in \cref{cor:polyCG}. Throughout, we use $[m]$ to denote the set $\{1,\dots,m\}$ and {$\text{Bin}(h,p)$ to denote a Binomial distribution with parameters $h$, $p$.}

\subsection{Background tools}
\label{subsec:ingredients}
\subsubsection{Label Cover}
 We start by introducing \gaplabelcover, a commonly utilized \NP-hard problem to obtain tight inapproximability results. We employ the weak-value formulation of the problem, implicitly used in \cite{Feige98} and also defined in \cite{DudyczMMS20}.
 {Informally, in a \gaplabelcover~problem, we are given a bipartite graph and a set of colors which we must use to color the left vertices of the graph. Additionally, we are given a function per each edge mapping the color chosen for the left incident vertex to a color for the right incident vertex.
The following paragraph formalizes the setting, while \cref{fig:instance_label_cover} provides an example. As common in the literature, we refer to the set of colors as to the alphabet, and to a color as to a label in the alphabet.} 

\medskip

\begin{definition}
\noindent A {\labelcover} instance is described by a tuple $(L,R,E,h,[\lalph],[\ralph],\{\pi_{e}\}_{e\in E})$, where 
 \begin{itemize}[leftmargin=*]
 \item[-]
$L$ and $R$ are sets of left and right vertices of a bi-regular bipartite graph with edge set $E$ and right degree $h$ (i.e., the degree of all vertices in $R$ equals $h$), 
\item[-]
$[\lalph]$ and $[\ralph]$ represent left and right alphabets, and
\item[-]
for every edge $e\in E$, a constraint function
$\pi_{e}:[\lalph]\rightarrow[\ralph]$ maps left labels to right labels. 
 \end{itemize}
\end{definition}

\medskip
\noindent
Given a left labeling $\lab:  L\rightarrow[\lalph]$, i.e., a map that associates every left vertex to a label, we say that a right vertex $u\in R$ is

\begin{itemize}[leftmargin=*]
\item[-] \emph{strongly satisfied} if for every pair of neighbors $v,v'\in L$ of $u$ it is $\pi_{(v,u)}(\lab(v))=\pi_{(v',u)}(\lab(v'))$;
\item[-] \emph{weakly satisfied} if there exist two distinct neighbors $v,v'\in L$ of $u$, \mbox{s.t. $\pi_{(v,u)}(\lab(v))\!=\pi_{(v',u)}(\lab(v'))$.} 	
\end{itemize}

\medskip
\noindent
{Stated differently, given an assignment of colors to the left vertices (i.e., a labeling), a right vertex $u\in R$ is weakly satisfied if there exists two edges incident to that vertex $u$ whose corresponding right color matches. If this property holds for all edges incident to $u$, then the right vertex is strongly satisfied. Depending on whether we can find a coloring of the left vertices satisfying all or a fraction of the right vertices, we then distinguish between a {\tt YES} and a {\tt NO} instance%
.}%

\medskip
\noindent

\begin{definition}
For any $\delta>0$, $h\in\mb{N}$ let \gaplabelcover$(\delta,h)$ denote the following problem: Given a \labelcover~instance $(L,R,E,h,[\lalph],[\ralph],\{\pi_{e}\}_{e\in E})$, distinguish between
\begin{itemize}[leftmargin=12.5mm]
\item[- {\tt YES}:] there exists a labeling that strongly satisfies all right vertices;	
\item[- {\tt~NO}:] no labeling weakly satisfies more than a fraction $\delta$ of the right vertices.
\end{itemize}
\end{definition}

\medskip
\noindent
{Distinguishing between a {\tt YES} and a {\tt NO} instance is known to be difficult (formally \NP-hard) when the size of the right alphabet, i.e., the number of possible right colors, is sufficiently large as recalled next. This result constitutes the basis for our hardness of approximation.
}

\begin{proposition}[{\it \cite{Feige98, DudyczMMS20}}]
\label{prop:hardness-gaplabel}
$\forall\delta>0$, $h\in\mb{N}$, $h\ge 2$, and $\ralph$ sufficiently large (depending on $\delta,h$), \gaplabelcover$(\delta,h)$ is \NP-hard.
\end{proposition}

\begin{figure}[t!]
\centering
\qquad
\begin{subfigure}[t]{0.3\textwidth}
         \centering
         \includegraphics[scale=0.4]{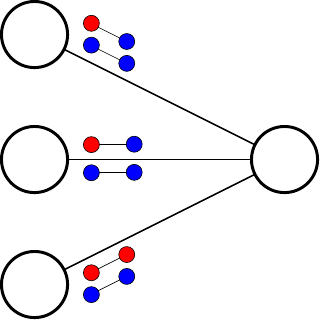}
         \caption{}
\end{subfigure}
\hfill
\begin{subfigure}[t]{0.3\textwidth}
         \centering
         \includegraphics[scale=0.4]{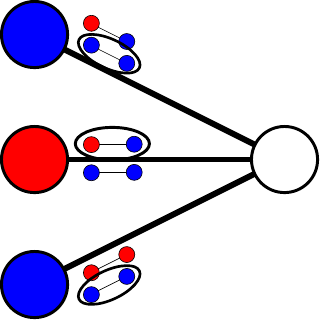}
         \caption{}
\end{subfigure}
\hfill
\begin{subfigure}[t]{0.3\textwidth}
         \centering
         \includegraphics[scale=0.4]{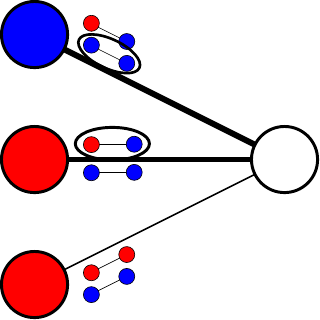}
         \caption{}
\end{subfigure}
\vspace*{2mm}
\caption{Left panel: instance of \labelcover~described by a bi-regular bipartite graph with $|L|=3$, $|R|=1$, $h=3$. Here, the left and right alphabets have equal size $\alpha=\beta=2$ and are represented by the colors $\{{\tt red},{\tt blue}\}$. The constraint functions $\pi_e$ are visualized over the corresponding edges, and map any possible choice of color for the left incident vertex into a color for the right incident vertex.
Central panel: instance of \labelcover~and coloring choice that strongly satisfies the right vertex. This is because the choice of left colors maps, through the edge constraints, to the same color on the right vertex.
Right panel: instance of \labelcover~and coloring choice that weakly (but not strongly) satisfies the right vertex. This is because there exists two left vertices (but not all) whose colors match, once mapped to the corresponding right vertex. 
}
\label{fig:instance_label_cover}
\end{figure}

\subsubsection{Partitioning System}
\label{subsubsec:partsys}
Similarly to \cite{barman2021tight} and \cite{DudyczMMS20}, we generalise a combinatorial object introduced by \cite{Feige98}, called \emph{partitioning system}, which we also equip with a cost function.
{
This object will be used to construct the \players{} allocation sets $\{\mc{A}_i\}_{i=1}^N$ in our ensuing reduction. In particular, it will ensure that, in our reduction, \players{} have allocations that are either optimal or incur a very high social cost, and the ratio between these two costs approaches the inapproximability factor we wish to achieve. A general illustration is provided in \cref{fig:partition_sys} while a concrete example is given in \cref{fig:partition_sys_small}.
}

\medskip
\begin{definition}
\label{def:partsys}
Given a ground set of elements $[n]$, integers $\ralph,h,\el$ such that $\el n/h\in\mb{N}$, $\ralph\ge h\ge \el$, a cost {function} $c:\mb{N}\rightarrow \mb{R}_{\mygezero}$, and $\eta>0$, a partitioning system with parameters $(n,\ralph,h,\el,\eta)$ is a collection of partitions $\mc{P}_1,\dots,\mc{P}_{\ralph}$ of $[n]$ such that:
\medskip
\begin{itemize}
\item[P1)] Every partition $\mc{P}_j$ is a collection of subsets $P_{j,1}, \dots, P_{j,h}\subseteq [n]$ each with $\el n/h$ elements and such that each element from $[n]$ is selected by $\el$ sets in $P_{j,1}, \dots, P_{j,h}$. Observe that, for any $\mc{P}_j$ we have $|\mc{P}_j|=h$, and the above implies
\be
\sum_{r\in [n]}c(|\mc{P}_j|_r)=c(\el) n.
\label{eq:rows_part}
\ee
Here, $|\mc{P}_j|_r$ denotes the number of sets in the collection $\mc{P}_j$ to which element $r$ belongs, where we extended the definition of the function $c$ to include $c(0)=0$. {This allows to simplify notation in \eqref{eq:rows_part} by summing over all resources $r\in [n]$, as opposed to summing over $r\in [n]$ s.t. $|\mc{P}_j|_r\ge 1$.}
\smallskip
\item[P2)] for any $\Tset\subseteq[\ralph]$ with $|\Tset|=h$ and for any function $i: [\ralph]\rightarrow[h]$, let $Q=\{P_{j,i(j)}, j \in \Tset\}$. It is
\be
\sum_{r\in [n]} c(|Q|_r)\ge \left(\mb{E}_{X\sim \text{Bin}(h,\el /h)}[c(X)]-\eta\right)n.
\label{eq:scrambled}
\ee
\end{itemize}
\end{definition}
\begin{figure}[t!]
\centering
\vspace*{2mm}
\includegraphics[scale=1.4]{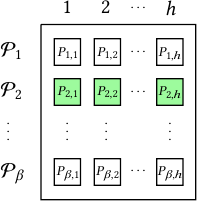}\qquad\qquad\qquad
\includegraphics[scale=1.4]{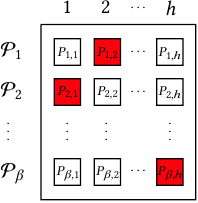}
\vspace*{2mm}
\caption{A partitioning system with parameters $(n,\ralph,h,k,\eta)$. Each box contains $kn/h$ elements from the ground set $[n]$. Property P1 ensures that selecting an entire row results in low cost (left, in green). Property P2 ensures that selecting one and only one box per each row, for a total of $h$ rows, results in a high cost (right, in red).}\vspace*{-2mm}
\label{fig:partition_sys}
\end{figure}
To gain some intuition on properties P1 and P2, we provide a graphical representation of a partitioning system in \cref{fig:partition_sys}, whereby each box contains $\el n/h$ elements from the ground set $[n]$. Property 1 asserts that every time we select {an} entire row, we are guaranteed to cover every element in $[n]$ precisely $\el$ times. Property 2 asserts that every time we select one and only one set from each row, for a total of $h$ rows, many sets cover the same \mbox{elements so that the resulting cost is high.} {A concrete example of partitioning system is provided in \cref{fig:partition_sys_small}.}
\begin{figure}[t!]
\centering
\qquad\qquad
\begin{subfigure}[t]{0.15\textwidth}
         \centering
         \includegraphics[scale=1.6]{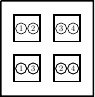}
         \caption{}
\end{subfigure}
\qquad
\begin{subfigure}[t]{0.3\textwidth}
         \centering
         \includegraphics[scale=1.6]{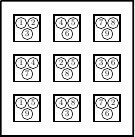}
         \caption{}
\end{subfigure}
\qquad
\begin{subfigure}[t]{0.3\textwidth}
         \centering
         \includegraphics[scale=1.6]{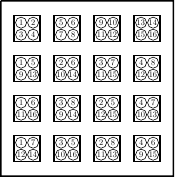}
         \caption{}
\end{subfigure}
\vspace*{2mm}
\caption{{Examples of partitioning systems obtained with the choice of $c(x)=x^2$, $k=1$, for increasing values of $\beta=h=2,3,4$ representing the number of rows and columns. Here, each labeled circle represents a resource from the ground set. To verify that (a) is indeed a partitioning system it suffices to observe that, when selecting any entire row, each resource is covered only once (since $k=1$), so that we incur a cost of $4\cdot c(1)=4$ (Property 1), while when selecting one box from each row we incur a cost of $c(1)+c(2)+c(1)=6=\mb{E}_{X\sim \text{Bin}(2,1/2)}[X^2]$ (Property 2). Analogously for cases (b) and (c).
Note how, in these examples, we have been able to construct partitioning systems where the error parameter $\eta$ entering \eqref{eq:scrambled} is identically zero. 
Finally, we observe that partitioning systems obtained with $c(x)=x^2$, $k=1$ correspond, in the ensuing reduction, to the special case of congestion games with linear resource cost $\ell(x)=x$.}}\vspace*{-2mm}
\label{fig:partition_sys_small}
\end{figure}

At this stage, we recall that the probability mass function of the binomial distribution $\text{Bin}(h,k/h)$ converges pointwise for fixed $k\ge h$ to the probability mass function of the Poisson distribution $\text{Poi}(k)$ as $h$ grows large \citep{durrett2019probability}%
. Hence, we informally observe that, when $\rho_\ell<\infty$, 
\[\frac
{\sum_{r\in Q} c(|Q|_r)}
{\sum_{r\in \mc{P}_j}c(|\mc{P}_j|_r)}
 = \frac{
\mb{E}_{X\sim \text{Bin}(h,\el /h) }[c(X)]-\eta}{c(\el)} \xrightarrow{h\rightarrow \infty}
\frac{
\mb{E}_{X\sim \text{Poi}(\el) }[c(X)]-\eta}{c(\el)},
\]
see \cref{lem:convergence} in \cref{app:binomial-poisson} for a proof. 
If we choose $c(x)=x\ell(x)$, let $h\rightarrow \infty$ and consider the worst case over $\el$, this ratio precisely matches the inapproximability result we aim to derive (cfr. the previous expression and $\rho_\ell$ in \eqref{eq:approxfact}). We are thus left to piece these elements together in the ensuing section. Before doing so, we remark that partitioning systems do exist for every choice of $\eta>0$ as long as $n$ is taken sufficiently large as stated in the next proposition. Its proof follows the same approach of that in \citep{barman2021tight}, and is included in \cref{subsub:app-proof-partition-exists} for completeness. 
We remark that, when used in the upcoming reduction, we will be able to compute a partitioning systems in time that is independent on the size of the instance we reduce from.

\begin{proposition}
\label{prop:partition-exists}
Let $c:\mb{N}\rightarrow \mb{R}_{\mygezero}$ non-decreasing be given. For every choice of $\beta\ge h\ge k$ integers with $kn/h\in\mb{N}$, $\eta\in(0,1)$, and $n\ge \frac{c(h)^2}{2\eta^2}[\log(10)+\ralph\log(h+1)]$ a partitioning system with parameters $(n,\beta,h,k,\eta)$ and cost function $c$ exists. It can be found in time depending solely on $h$, $n$ and $\ralph$.
\end{proposition}}
\subsection{Reduction}
\label{subsec:reduction}
We first provide the reduction and prove the hardness result in the case when $\rho_\ell<\infty$. We consider the case $\rho_\ell=\infty$ separately at the end of \cref{subsec:actualproof}. 
Starting from the resource cost function $\ell$ and a fixed $\varepsilon>0$, we will first construct a partitioning system with {appropriately defined parameters}. We then reduce an instance of $\labelcover$ $\mc{C}=(L,R,E,h,[\lalph],[\ralph],\{\pi_{e}\}_{e\in E})$ to an instance of congestion game $G=(N,\mc{R},\{\mc{A}_i\}_{i=1}^N,\{\ell_r\}_{r\in\mc{R}})$ with identical resource cost.
The idea is to define $G$ by creating a copy of the partitioning system for every right vertex and use that to \mbox{define the players' allocation sets $\mc{A}_i$.}

\medskip
{{\it Construction of the partitioning system.} Given $\ell$ and
$\varepsilon>0$, we construct a partitioning system setting its parameters as follows:
\begin{itemize}
	\item[-] Let $\el$ be the maximizer  of \eqref{eq:approxfact};\footnote{For ease of exposition, we show the result when the supremum is attained at some value $k\in \mb{N}$. If this is not the case, then the supremum must be achieved at $k\rightarrow\infty$. One then fixes $k\in\mb{N}$ and proceeds with the reasoning as is. This will give rise to an additional error term $\nu(k)$ with $\lim_{k\to\infty} \nu(k)=0$ in the ensuing Equation \eqref{eq:limit}, where the right hand side will be replaced by $(\rho_\ell+\nu(k))c(k)$. Nevertheless one can select $k$ sufficiently large and control such error to a desired accuracy.}%
\item[-] Let $c$ be the cost function defined from $\ell$ as in $c(x)=x\ell(x)$ for all $x\in\mb{N}$ and $c(0)=0$;
\item[-] Let $\eta< \min\{\varepsilon c(k) /4,1\}$; 
\item[-] Let $h\ge k$ integer be such that $|\mb{E}_{X\sim \text{Bin}(h,\el/h)} [X\ell(X)] -\mb{E}_{X\sim \text{Poi}(\el)} [X\ell(X)]|\le \varepsilon c(k)/4$, which exists thanks to the convergence result in \cref{lem:convergence} in \cref{app:binomial-poisson};
\item[-] \mbox{Let $\delta\le \varepsilon/(2\rho_\ell)$ and fix $\ralph$  large enough to ensure hardness of \gaplabelcover$(\delta,h)$ in \cref{prop:hardness-gaplabel};}
\item[-] Let $n$ be so that $kn/h\in\mb{N}$ and $n$ is large enough to have existence, from \cref{prop:partition-exists}, of the partitioning system with parameters $(n,\ralph,h,k,\eta)$ and cost function $c$. 
\end{itemize}}
\medskip
 Owing to \cref{prop:hardness-gaplabel}, since $\ralph,h,k,\eta$ are now fixed, a partitioning system can be computed in time {\it independent} of the size of the $\labelcover$ instance we wish to reduce from. 
  
  \medskip
  {{\it Reduction of \labelcover~instance to congestion game instance.}}
  Now we take an instance of $\labelcover$ $\mc{C}=(L,R,E,h,[\lalph],[\ralph],\{\pi_{e}\}_{e\in E})$ where $h$, $\ralph$ are defined above and thus 
      \cref{prop:hardness-gaplabel} (\NP-hardness) holds.
  For each right vertex $u\in R$ we use the local partitioning system with parameters $(n,\ralph,h,\el,\eta)$. We refer to the resources in the partitioning system corresponding to the right vertex $u\in R$ with $\{1^u,\dots,n^u\}$. Similarly we use $\mc{P}^u_j=\{{P}^u_{j,1},\dots, {P}^u_{j,h}\}$ for the local partitions. 
\mbox{The congestion game $G=(N,\mc{R},\{\mc{A}_i\}_{i=1}^N,\{\ell_r\}_{r\in\mc{R}})$ is defined as follows} 
\begin{itemize}
	\item[-] each left vertex corresponds to an \player{}, so that the number of \players{} is $N=|L|$;
	\item[-] the ground set of resources is the union of the resources introduced by each local partitioning system on every right vertex, i.e., $\mc{R}=\cup_{u\in R}\{1^u,\dots,n^u\}$;
	\item[-] each resource cost is equal to $\ell$, i.e., $\ell_r(x)=\ell(x)$ for all $r\in\mc{R}$, $x\in\mb{N}$;
	\item[-] as each left vertex $v\in L$ corresponds to one and only one \player{} $i\in [N]$, we refer to a left vertex as $i\in [N]$ instead of as $v\in L$ to ease the notation.	For \player{} $i\in [N]$ we construct each pure strategy $a_i\in\mc{A}_i$ as follows.  We let the left vertex $i$ select a label $\llabel\in[\lalph]$, and correspondingly take the union over all right vertices $u\in \mc{N}(i)$ neighbouring with $i$, of the resources belonging to the block $P^u_{j,i}$, where $j=\pi_{(i,u)} (\llabel)$. Repeating over all left labels we \mbox{obtain the strategy set $\mc{A}_i$. Formally}
	\[
	\mc{A}_i=\left\{\cup_{u\in \mc{N}(i)} {P}^u_{j,i}, \quad\text{where} \quad j=\pi_{(i,u)}(\llabel),\quad  \forall \llabel\in [\lalph] 
	\right\}.
	\]
\end{itemize}
	The following figure exemplifies the construction. We conclude remarking that the above procedure implicitly defines a map associating a profile of left labels $(\llabel_1,\dots,\llabel_N)$ (one per each left vertex) to an allocation $(a_1,\dots,a_N)\in\mc{A}$, and that spanning through all possible choices of $(\llabel_1,\dots,\llabel_N)$ produces all possible allocations in $\mc{A}$. This observation will be useful in proving the hardness result.
\begin{figure}[htb!]
\vspace*{3mm}
\centering
\includegraphics[scale=1.3]{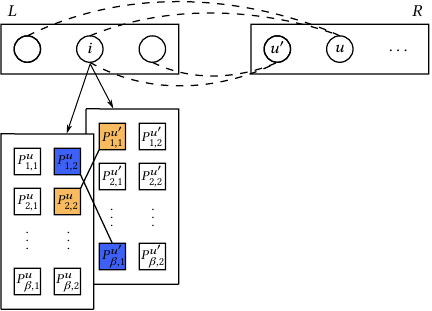}
\vspace*{3mm}
\caption{Given a label cover instance $\mc{C}=(L,R,E,h,[\lalph],[\ralph],\{\pi_{e}\}_{e\in E})$, our reduction associates every left vertex in $L$ to an \player{} in the game $G$. Here we exemplify how the action set $\mc{A}_i$ is generated for  \player{} $i\in L$. To ease the presentation, we consider a left alphabet of size $2$ and use {\tt blue} and {\tt orange} to identify the left labels. Since $i$ has two right neighbours, $u$ and $u'$, we construct two partitioning systems with ground set of resources $\{1^u,\dots,n^u\}$ and $\{1^{u'},\dots,n^{u'}\}$. Constraints $\pi_{(i,u)}({\tt blue})=1$, $\pi_{(i,u)}({\tt orange})=2$ and $\pi_{(i,u')}({\tt blue})=\ralph$, $\pi_{(i,u')}({\tt orange})=1$ are given, and we represent them graphically with the fact that on the left partitioning system, the label ${\tt blue}$ (resp. ${\tt orange}$)  is associated to the block in row $1=\pi_{(i,u)}({\tt blue})$ (resp. row $2=\pi_{(i,u)}({\tt orange})$). Similarly for the right partitioning system using $\pi_{(i,u')}$ to determine the row.
The set $\mc{A}_i$ is readily constructed as $\mc{A}_i=\{P_{1,2}^u\cup P^{u'}_{\ralph,1}, P_{2,2}^u\cup P^{u'}_{1,1}\}$, where the first (resp. second) allocation corresponds to a {\tt blue}  (resp. {\tt orange}) left label.  
}
\end{figure}

\subsection{Proof of the result}
\label{subsec:actualproof}
As anticipated, we first prove the result for the case of $\rho_\ell<\infty$. At the end of this section, we turn the attention to $\rho_\ell=\infty$. For any given instance of $\labelcover$ $\mc{C}=(L,R,E,h,[\lalph],[\ralph],\{\pi_{e}\}_{e\in E})$, resource cost $\ell$, and $\varepsilon>0$, 
we consider an instance of congestion game $G=(N,\mc{R},\{\mc{A}_i\}_{i=1}^N,\{\ell_r\}_{r\in\mc{R}})$ constructed as in the previous section.
We will now show that%
\begin{itemize}
\item[-] {\it completeness (\cref{subsec:completeness})}: If the instance $\mc{C}$ is a {\tt YES}, then $\min_{a\in \mc{A}}SC(a)\le n|R|c(k)$,	  
\item[-] {\it soundness (\cref{subsec:soundness})}: If the instance $\mc{C}$ is a {\tt NO}, then 	$\min_{a\in \mc{A}}SC(a)> (\rho_\ell-\varepsilon)n|R|c(k)$.
\end{itemize} 
An algorithm solving $\mincon$ with an approximation ratio smaller than {$\rho_\ell-\varepsilon$ will be able to distinguish between {\tt YES}/{\tt NO} instances of an \NP-hard \mbox{promise problem. This will conclude the proof.}

\subsubsection{Completeness}
\label{subsec:completeness}
We intend to show that if $\mc{C}$ is a {\tt YES} instance, then $\min_{a\in \mc{A}} \SC(a) \le n |R| c(\el)$. This follows readily. In fact, if $\mc{C}$ is a {\tt YES} instance, there exists a labeling that strongly satisfies all right vertices. This implies that there exists an allocation $a^\star\in \mc{A}$ whereby, for any given right vertex, all neighbouring left vertices (\players{}) have selected blocks belonging to an entire row of the corresponding partitioning system. Thanks to property P1 the cost of the allocation $a^\star$ on every local partitioning system is equal to $n c(\el)$. Since the total cost is additive over the local partitioning systems, we obtain the result
\[
SC(a^\star)= n |R| c(\el) \implies \min_{a\in \mc{A}}SC(a)\le n |R| c(\el).
\]

\subsubsection{Soundness}
\label{subsec:soundness}
We intend to show that if $\mc{C}$ is a {\tt NO} instance, then $\min_{a\in \mc{A}}\SC(a)\ge (\rho_\ell-\varepsilon) n|R|c(k)$ which is equivalent to showing  $\SC(a)\ge (\rho_\ell-\varepsilon) n|R|c(k)$ for all $a\in \mc{A}$. Towards this goal, we build upon the last observation presented in \cref{subsec:reduction}, i.e., the fact that our construction associates each profile of left labels to an allocation, and that spanning through all possible choices of $(\llabel_1,\dots,\llabel_N)$ produces all possible allocations in $\mc{A}$. 
Hence, it suffices to prove the {desired} property by considering all possible combinations of profiles $(l_1,\dots,l_N)$ and the corresponding induced cost, instead of considering all $a\in \mc{A}$. Since $\mc{C}$ is a {\tt NO} instance, for any possible choice of $(l_1,\dots,l_N)$, no more than $\delta$ fraction of the right vertices are weakly satisfied. Owing to property P2, each partitioning system corresponding to a non weakly satisfied right vertex has a cost larger than $\left(\mb{E}_{X\sim \text{Bin}(h,\el /h)}[c(X)]-\eta\right)n$. Thus, since at least $(1-\delta)|R|$ right vertices \mbox{are not weakly satisfied, it is}
\[
\SC(a)\ge \left(\mb{E}_{X\sim \text{Bin}(h,\el /h)}[c(X)]-\eta\right)(1-\delta)n|R|,\qquad\forall a\in\mc{A}.
\]
We conclude with some cosmetic manipulation. In particular, we recall that the binomial distribution converges to the Poisson distribution when the number of trials grows large and the success probability of each trial goes to zero \citep{durrett2019probability}. In our settings, this corresponds to the fact that the probability mass function of $\text{Bin}(h,\el/h)$ converges pointwise for fixed $k$ to that of $\text{Poi}(\el)$ as $h\rightarrow\infty$. 

Since $\rho_\ell<\infty$, we observe that 
\be
\label{eq:limit}
\lim_{h\to\infty}\mb{E}_{X\sim \text{Bin}(h,\el/h)} [c(X)] = 
\mb{E}_{X\sim \text{Poi}(\el)} [c(X)] = \rho_\ell c(\el),
\ee
where equality holds thanks to \cref{lem:convergence}  in \cref{app:binomial-poisson}. %
This implies the existence of a function $\theta(h)$ with $\theta(h)\rightarrow 0$ as $h\rightarrow \infty$ allowing to control the error, and for which 
$
 \mb{E}_{X\sim \text{Bin}(h,\el/h)} [c(X)]
\ge \rho_\ell c(\el)-\theta(h). 
$
In other words the LHS can be made arbitrarily close to $\rho_\ell c(\el)$ by selecting $h$ sufficiently large. Hence, 
\be
\label{eq:last-step}
\begin{split}
\SC(a)&\ge
(\rho_\ell c(\el)-\theta(h) - \eta)(1-\delta)n|R|\\
&=\left[\rho_\ell
- \frac{\theta(h)}{c(\el)}
-\frac{\eta}{c(\el)}
- \left(\rho_\ell - \frac{\theta(h)}{c(\el)}-\frac{\eta}{c(\el)}\right)\delta \right] n |R| c(\el)\\
&\ge 
\left[\rho_\ell
- \frac{\theta(h)}{c(\el)}
-\frac{\eta}{c(\el)}
-\rho_\ell\delta \right] n |R| c(\el)\\
&>\left[\rho_\ell- \frac{\varepsilon}{4}- \frac{\varepsilon}{4} -  \frac{\varepsilon}{2}\right]n |R| c(\el) = (\rho_\ell-\varepsilon)n|R|c_k.
\end{split}
\ee
The last inequality holds by the choice of parameters, ensuring that $\frac{\theta(h)}{c(k)}\le \frac{\varepsilon}{4}$, \mbox{$\frac{\eta}{c(k)}< \frac{\varepsilon}{4}$, $\rho_\ell\delta\le \frac{\varepsilon}{2}$.}

\paragraph{Case of $\rho_\ell=\infty$.} As anticipated, we treat the case of unbounded $\rho_\ell$ separately. Towards this goal, we follow the same reduction of \cref{subsec:reduction}, with minor modification on the choice of parameters. We replace $\rho_\ell$ with a fixed (and conceptually large) constant $M$. 
Since $\rho_\ell=\infty$, we note that $\mb{E}_{X\sim \text{Poi}(\el)} [c(X)]/c(k)$ is unbounded at some $k$ (possibly infinity). Since the probability mass functions of $\text{Bin}(h,\el/h)$ and $\text{Poi}(\el)$ converge, we can choose the pair $h$ and $k$ so that $\mb{E}_{X\sim \text{Bin}(h,\el/h)} [c(X)]\ge M c(k)$. Finally, we set $\delta\le \varepsilon/(2M)$. One then follows the same proof as in the case of bounded $\rho_\ell$, whereby \eqref{eq:last-step} is replaced with $\SC(a)\ge (M-\eta/c(k))(1-\delta)n|R|c(k)\ge 
(M
-\frac{\eta}{c(\el)}
-M\delta) n |R| c(\el)\ge (M-\varepsilon/4-\varepsilon/2)n|R|c(k)>(M-\varepsilon)n|R|c(k)$. Since $M$ can be taken to be arbitrarily large, the problem is \NP-hard to approximate within any finite ratio.
\hfill\qed

\subsection{Hardness factor for polynomial resource cost}
\label{subsec:bell-number}

\cref{cor:polyCG} claims that, when resource costs are obtained by non-negative combinations of monomials of maximum degree $d>0$, \mincon~is hard to approximate within any factor smaller than the $(d+1)$'st Bell number. Note that this is a direct consequence of \cref{thm:hardness} and its ensuing discussion, which applies since for $d\ge0$, each monomial $x^d$ is positive, non-decreasing, semi-convex for $x\in\mb{N}$.  {
In this setting, \mincon~is \NP-hard to approximate within any factor smaller than $\sup_{\ell\in\mathscr{L}}\rho_{\ell}$, where $\mathscr{L}$ contains all polynomials of maximum degree $d$ with non-negative coefficients. Interestingly, \Cref{lem:factor_polynomials} in \Cref{app:factor-poly} shows that the factor $\rho_\ell$ is maximized amongst all polynomials in $\mathscr{L}$ by the monomial $m(x)=x^d$. Hence, for polynomial congestion games, it is
}
\[
\sup_{\ell\in\mathscr{L}} \rho_\ell=\rho_m=\sup_{x\in\mb{N}}\frac{\mb{E}_{P\sim \text{Poi}(x)}[P^{d+1}]}{x^{d+1}}=\sup_{x\in\mb{N}}~~\sum_{i=0}^{d+1} x^{i-(d+1)}{d+1 \brace i}=\sum_{i=0}^{d+1}{d+1 \brace i}=\mc{B}(d+1).
\]
{The first equality follows from the fact that $m(x)=x^d$ maximizes the ratio $\rho_\ell$, while the second from the definition of $\rho_m$.}
The third equality follows from the fact that the $(d+1)$'st moment of the Poisson distribution $\text{Poi}(x)$ equals $\sum_{i=0}^{d+1}x^i {d+1 \brace i}$, where ${d+1 \brace i}$ is a Stirling number of the second kind \cite[p. 63]{mansour2015commutation}. The fourth equality holds because each function $x^{i-(d+1)}$ is non-increasing, owing to $i\le (d+1)$, and thus the supremum is attained at $x=1$. The last one is due to the definition of the $(d+1)$'st Bell number, which we denote with $\mc{B}(d+1)$, see \cite[Eq. 1.2]{mansour2015commutation}. 

One can repeat a similar reasoning also when $d\ge0$ is not integer and show that the expression inside the supremum is non-increasing in $x\ge1$, e.g., by computing its derivatives. If follows that the supremum is attained at $x=1$, and the definition of $\rho_\ell$ gives
\be
\label{eq:dobinski}
\rho_\ell=\frac{1}{e}\sum_{i=0}^\infty \frac{i^{d+1}}{i!}.
\ee
The latter expression is sometimes referred to as the fractional Bell number. Note that using Dobi\'nski's formula \cite[Eq. 1.25]{mansour2015commutation}, one recovers the Bell number $\mc{B}(d+1)$ when $d\in\mb{N}_0$. %

\medskip
{\cref{cor:affine_CG_refined} provides a more refined analysis for affine congestion games where $\ell_r(x)=\mya_rx+\myb_r$, when an upper bound $q$ on $\max_r \mya_r/(\mya_r+\myb_r)$ is given. The proof of \cref{cor:affine_CG_refined} follows readily from the application of \cref{thm:hardness} and its ensuing discussion. Indeed, here one can take the set $\mathscr{L}$ to contain all linear cost functions of the form $\ell(x)=\mya x+\myb$ where $\mya/(\mya+\myb)\leq q$. One can then easily compute $\rho_\ell$ explicitly from \eqref{eq:approxfact}
\[
\rho_\ell = \sup_{x\in\mb{N}}\frac{\mb{E}_{P\sim \text{Poi}(x)}[P(\mya P+\myb)]}{x(\mya x+\myb)}=\sup_{x\in\mb{N}}\frac{\mya (x^2+x)+\myb x}{\mya x^2+\myb x}
=1+\frac{\mya}{\mya+\myb},
\] and thus obtain the desired hardness factor, i.e., $\sup_{\ell\in\mathscr{L}}\rho_\ell=1+q$.}

%% file: parts-body/opttaxes.tex
\section{Taxes achieve optimal approximation}
\label{sec:taxes}
In this section we show how to compute a taxation mechanism whose price of anarchy matches the hardness factor. Since taxation mechanisms can be utilized to derive polynomial time algorithms with an approximation factor matching the price of anarchy (\cref{subsec:contributions}), in the ensuing \cref{subsec:comparison-sviri} we compare the optimal price of anarchy with the best known polynomial approximation of \cite{Makarychev18}. 

\medskip
We start by introducing a parameterised family of taxation mechanisms for which we will later provide (efficiently computable) parameters that achieve the desired result. Our taxation mechanisms take as input a congestion game $G$, with resource costs belonging to a common set $\mathscr{L}$. %
	{Toward this goal, for each $\ell_r:\mb{N}\rightarrow\mb{R}_{\mygezero}$ given, we extend its definition to $\ell_r:\mb{N}_0\rightarrow\mb{R}_{\ge 0}$ by setting $\ell_r(0)=0$. This is without loss of generality and merely needed to ease the notation (see \cpageref{foot:simplify_b(0)} and \cref{foot:simplify_b(0)}). For similar reasons, we define $c_r:\mb{N}_0\rightarrow\mb{R}_{\ge 0}$ as $c_r(x)=x\ell_r(x)$ for $x\in\mb{N}$ and $c_r(0)=0$.} Finally, we associate each resource cost $\ell_r(x)$ with the function $\hfun_r:\mb{R}_{\ge0}\rightarrow \mb{R}_{\ge0}$ defined as 
	\be
	\label{eq:binExp}
    \hfun_r(v)= \mb{E}_{P\sim \text{\normalfont Poi}(v)}[P\ell_r(P)] = \left(\sum_{i=0}^\infty i\ell_r(i)\frac{v^i}{i!}\right)e^{-v},
	\ee
	{which we use to introduce the following set of parametrized mechanisms.}

\begin{definition}[Parameterised Taxation Mechanisms]	
\label{def:family}
	Given a parameter vector $(v_r)_{r\in \mc{R}}$ with $v_r\in\mb{R}_{\ge0}$ and a set of resources $\mc{R}$ with costs $\ell_r(x)$%
	, define a parameterised taxation function 
	$\tau_r:\mb{N}_0\times\mb{R}_{\ge0} \rightarrow \mb{R}_{\ge0}$, with $\tau_r(x,v_r)=f_r(x,v_r)-\ell_r(x)$ %
	where $f_r(x,0)=\ell_r(x)$, $f_r(0,v)=0$, and
\be
\label{eq:f-designed}
f_r(x,v)=\frac{(x-1)!}{v^x}\sum_{i=0}^{x-1} \frac{\hfun_r(v)-i\ell_r(i)}{i!}{v^i},\qquad x\in\mb{N}_0, v\in\mb{R}_{>0}.
\ee
\end{definition}
\noindent{Observe that the previous taxation mechanism effectively substitutes the original resource costs $\ell_r(x)$ with the new resource costs $f_r(x,v_r)$ given in \eqref{eq:f-designed}. Indeed, when taxes are factored in, the resource cost perceived by each \player{} on resource $r$ is $\ell_r(x)+\tau_r(x,v_r)=f_r(x,v_r)$.}
{The following lemma provides three important properties of the above class of taxation mechanisms. In particular, \cref{lem:propertiesTM} ensures that taxes are non-negative, that modified resource costs introduced in \eqref{eq:f-designed} are non-decreasing and that they satisfy a recursion which will be crucial later on. The proof can be found in \cref{app:twoproperties}.}

{
\begin{lemma}
\label{lem:propertiesTM}
For all $v\in\mb{R}_{\ge 0}$,  $x\in \mb{N}_0$, and $r\in\mc{R}$ the taxation mechanism introduced in \cref{def:family} satisfies:
\begin{itemize}
	\item [(a)] $\tau_r(x,v)\ge 0$,
	\item [(b)] $f_r(x+1,v) \ge f_r(x,v)$,
	\item [(c)] $x \ell_r(x)-x f_r(x,v)+v f_r(x+1,v)=\hfun_r(v)$.
\end{itemize}
\end{lemma}
}

{As we will see in the remainder of the section, the mechanism that optimises the price of anarchy makes use of the parametrized taxation mechanism introduced in \cref{def:family}, with parameters obtained by solving the following convex program in the variables $y_i$, $v_r$ 
\be
\begin{aligned}
\label{eq:cvxprogram}
\min%
\quad & \sum_{r\in\mc{R}}
\hfun_r(v_r)\\
\text{\normalfont subject to}\quad & v_r = \sum_{i=1}^N\sum_{k\in[s_i]\,:\,r\in a_{i,k}} y_{i,k}\quad && \text{ for all } r\in\mc{R}, \\
     & y_i\in \Delta(s_i)&& \text{ for all } i\in[N], \\
\end{aligned}
\ee
where we let $s_i=|\mc{A}_i|$, $a_{i,k}$, denote the $k$-th action available to player $i$, and $\Delta(s)$ represent the $s$-th dimensional simplex. 
{This program can be easily interpreted as the continuous relaxation of the original $\mincon$ problem, where the resource costs $\ell_r$ have been replaced by the modified resource costs $p_r$ previously defined.}\\
We are now ready to state our main result of this section, which is an extension of \cref{thm:taxes_poa}. We state the result when $\sup_{\ell\in\mathscr{L}} \rho_{\ell}<\infty$, else the problem is inapproximable as seen in~\cref{thm:hardness}.
\medskip
\begin{theorem}
\label{thm:opttaxes}
 
Consider a congestion game where each resource cost  belongs to a common set of functions $\mathscr{L}$ {satisfying Assumption~\ref{ass:assumption}}. %
 Let $\sup_{\ell\in\mathscr{L}} \rho_{\ell}<\infty$ and denote with  %
$(\bar{y}_i)_{i\in [N]}, (\bar{v}_r)_{r\in\mc{R}}$
 a solution of the convex program \eqref{eq:cvxprogram}. Then:
\begin{itemize}
\item The taxation mechanism introduced in Definition \ref{def:family} with parameter vector $(\bar{v}_r)_{r\in \mc{R}}$ has a price of anarchy no-higher than $\sup_{\ell\in\mathscr{L}} \rho_{\ell}$;
\item Moreover, for any choice of $\varepsilon>0$ one can design in polynomial time, through the approximate solution of \eqref{eq:cvxprogram}, a taxation mechanism whose price of anarchy is no-higher than $\sup_{\ell\in\mathscr{L}} \rho_{\ell}+\varepsilon$.
\end{itemize}
\end{theorem}
\medskip

\begin{proof}{Proof.}
Given a congestion game $G$, we consider the corresponding program \eqref{eq:cvxprogram}. Let us verify that \eqref{eq:cvxprogram} is indeed convex. Since the constraints are linear and the objective function is a sum of univariate functions 
    $\hfun_r$ it suffices to show that each $\hfun_r(v)$ is convex in $v$. This holds true as $\hfun_r$ is defined in \eqref{eq:binExp} as the expectation of a convex function over a Poisson distributed random variable. For completeness we provide a proof of this fact in \Cref{lem:pois-is-cvx} in \cref{app:pois-cvx}.

  \smallskip
    
Let %
$(\bar{y}_i)_{i\in [N]}, (\bar{v}_r)_{r\in\mc{R}}$
be an optimal solution of the convex program \eqref{eq:cvxprogram} 
and consider the taxation mechanism from Definition \ref{def:family} with parameter vector $(\bar{v}_r)_{r\in \mc{R}}$. Denote $\rho=\sup_{\ell\in\mathscr{L}} \rho_{\ell}$.

To complete the proof we will use a smoothness approach, with a crucial modification: Instead of comparing an action profile $a$ (e.g., an equilibrium allocation) with another action profile $a'$ (e.g., an optimal allocation), we will compare an action profile $a$ against a mixed profile $y=(y_1,\dots,y_N)\in\Delta(s_1)\times\dots\times \Delta(s_N)$. Specifically, we will choose the mixed profile $\bar{y}$ solving \eqref{eq:cvxprogram}, and show that 
\be
\sum_{i=1}^N\sum_{k=1}^{s_i} \bar{y}_{i,k}[\bar{C}_i(a)-\bar{C}_i(a'_{i,k}, a_{-i})] \ge \SC(a)-\rho\SC(\opt{a}), \qquad \forall a\in\mc{A}.
\label{eq:smoothness}
\ee
where $\bar{C}_i(a)$ denotes the modified cost function perceived by the \players{} $\bar{C}_i(a)=\sum_{r\in a_i}f_r(|a|_r,\bar{v}_r)$. Once \eqref{eq:smoothness} is shown, the desired bound on the price of anarchy follows readily for pure Nash equilibria and more generally extends all the way to coarse correlated equilibria \citep{roughgarden2015intrinsic}. 
In the former case, substituting the profile $a$ with any pure Nash equilibrium $\NE{a}$, and summing the equilibrium conditions $0\ge \bar{C}_i(\NE{a})-\bar{C}_i(a'_{i,k}, \NE{a}_{-i})$, one obtains  $0\ge \sum_{i=1}^N\sum_{k=1}^{s_i} \bar{y}_{i,k}[\bar{C}_i(\NE{a})-\bar{C}_i(a'_{i,k}, \NE{a}_{-i})]$, so that
\[
0\ge \sum_{i=1}^N\sum_{k=1}^{s_i} \bar{y}_{i,k}[\bar{C}_i(\NE{a})-\bar{C}_i(a'_{i,k}, \NE{a}_{-i})] \ge \SC(\NE{a})-\rho\SC(\opt{a}),%
\]
from which one concludes. Since \eqref{eq:smoothness} holds for all $a\in\mc{A}$, the same bound on the price of anarchy holds for the much broader class of coarse correlated equilibria. To see this, let $\sigma$ be any coarse correlated equilibrium over $\mc{A}_1\times\dots\times\mc{A}_N$, and consider the expected value of \eqref{eq:smoothness}. Due to linearity of the expectation and the definition of coarse correlated equilibria, we have $0\ge \mb{E}_{a\sim\sigma}\left[\sum_{i=1}^N\sum_{k=1}^{s_i} \bar{y}_{i,k}[\bar{C}_i(a)-\bar{C}_i(a'_{i,k}, a_{-i})]\right]$, from which one concludes.

\smallskip
We are thus left to prove the smoothness condition \eqref{eq:smoothness}. Given an optimal allocation $\opt{a}\in\argmin_{a\in\mc{A}}\SC(a)$, let $\opt{v}_r=|\opt{a}|_r$. Inequality \eqref{eq:smoothness} follows from 
\be
\label{eq:finalstep}
\begin{aligned}
\sum_{i=1}^N\sum_{k=1}^{s_i} \bar{y}_{i,k}[\bar{C}_i(a)-\bar{C}_i(a'_{i,k}, a_{-i})]
&=\sum_{i=1}^N \bar{C}_i(a)
 -\sum_{i=1}^N\sum_{k=1}^{s_i}\bar{y}_{i,k}\bar{C}_i(a'_{i,k}, a_{-i})\\
(\text{$\bar{\ell}_r$ is non-decreasing by \Cref{lem:propertiesTM}-b})\quad 
&\ge \sum_{i=1}^N \bar{C}_i(a)
 -\sum_{i=1}^N\sum_{k=1}^{s_i}\bar{y}_{i,k} \sum_{r\in a'_{i,k}} f_r(|a|_r+1,\bar{v}_r)\\
{\text{(changing order of summation)}}\quad 
&= \sum_{r\in a} \left[|a|_r f_r(|a|_r,\bar{v}_r)-\bar{v}_rf_r(|a|_r+1,\bar{v}_r)\right]\\
\text{(recursion in \Cref{lem:propertiesTM}-c)}\quad 
&=\sum_{r\in a}\left[c_r(|a|_r)-\hfun_r(\bar{v}_r)\right]\\
\text{($\bar{v}_r$ optimal solution of \eqref{eq:cvxprogram})}\quad 
&\ge \sum_{r\in a}[c_r(|a|_r)-\hfun_r(\opt{v}_r)]\\
\text{(by def. of $\rho$, it is $\hfun_r(\opt{v}_r)\le \rho \,c_r(\opt{v}_r)$)}\quad
&\ge \sum_{r\in a}[c_r(|a|_r)-\rho c_r(\opt{v}_r)]\\
&= \SC(a)-\rho \SC(\opt{a}),
\end{aligned}
\ee
{thus completing the proof of the first claim.

We now turn attention to the second claim.}
Since \eqref{eq:cvxprogram} is a convex program with polynomially many decision variables and constraints, it can be solved to arbitrary precision in polynomial time. In this case, the argument in \cref{eq:finalstep} will go through with a minor change on the fifth line, where one pays a multiplicative factor $1+\xi$ with $\xi>0$. Correspondingly, we obtain a price of anarchy of $(1+\xi)\rho$ in place of $\rho$. Selecting $\xi$ sufficiently small, one obtains a price of anarchy of $(1+\xi)\rho \le \rho+\varepsilon$ for any choice of $\varepsilon>0$.
\hfill\qed
\end{proof}

{We conclude observing that the expression of $\hfun_r(v)$ in \eqref{eq:binExp} can be computed analytically for many commonly studied classes of resource costs. Nevertheless, if this is not the case, one can approximate $\hfun_r(v)$ to arbitrary precision by truncating the sum at a finite value. One can then verify the same properties shown in  \cref{lem:propertiesTM} and apply a reasoning identical to that in \eqref{eq:finalstep}. An additional multiplicative error will arise in \eqref{eq:finalstep}, although this can be made \mbox{as small as desired.}}

\subsection{Comparison with existing approximations}
\label{subsec:comparison-sviri}
A strength of the approach followed thus far is that optimally designed taxes can be used to derive polynomial time algorithms matching the hardness factor. This can be done relying on existing algorithms, e.g., no-regret dynamics, as discussed in \cref{subsec:contributions}. For this reason, we now turn the attention to comparing the optimal price of anarchy of \cref{thm:opttaxes} with the best known approximation ratio of \cite{Makarychev18}. %
Specifically, when all resource costs are identical to $\ell:\mb{R}_{\ge0}\rightarrow \mb{R}_{\ge0}$, \cite{Makarychev18} give a randomized algorithm with an approximation ratio of $\mu_\ell = \sup_{x\in \mb{R}_{>0}} {\mb{E}_{P\sim\text{Poi}(1)} [(xP) \ell(xP)]}/{(x\ell(x))}$. While their result applies to the broader class of optimization problems with a ``diseconomy of scale'', the approximation ratio in \cref{cor:polyalgo} always matches or strictly improves upon theirs. This follows from
\[
\rho_\ell = 
\sup_{x\in\mb{N}} \frac{\mb{E}_{P\sim\text{Poi}(x)}[P\ell({P})]}{x\ell(x)}
\le  
\sup_{x\in\mb{N}} \frac{\mb{E}_{P\sim\text{Poi}(1)} [(xP) \ell(xP)]}{x\ell(x)}
\le  
\mu_\ell.
\] 
The last inequality follows trivially by replacing $\mb{N}$ with $\mb{R}_{>0}$ and using the definition of $\mu_\ell$, while the first inequality can be shown using the notion of convex ordering between distributions.\footnote{To see this, we leverage the result in \citep[Thm 3.A.36]{shaked2007stochastic} with $X_i,Y\sim x\text{Poi}(1)$ independent,  $a_i=1/x$, and $i=1,\dots,x\in\mb{N}$. Theorem 3.A.36 in \citep{shaked2007stochastic} applies  %
ensuring that $\sum_{i=1}^x a_i X_i\le_{\text{cx}}Y$, where $\le_{\text{cx}}$ denotes the convex ordering between distributions. Thanks to the choice of $X_i,Y,a_i$, and to the fact that the sum of $x$ independent Poisson random variables with distribution $\text{Poi}(1)$ is itself a Poisson random variable with distribution $\text{Poi}(x)$, the statement $\sum_{i=1}^x a_i X_i\le_{\text{cx}}Y$ reads as $\text{Poi}(x)\le_{\text{cx}} x\text{Poi}(1)$ which implies $\mb{E}_{P\sim\text{Poi}(x)}[P\ell({P})] \le \mb{E}_{P\sim\text{Poi}(1)} [(xP)\ell(xP)]$, as $x\ell(x)$ is convex.

} %
An example where the approximation ratios coincide is given by $\ell(x)=x^d$, in which case $\rho_\ell=\mu_\ell$ equal the $(d+1)$'st Bell number, while an instance where the inequality is strict is provided by $\ell(x)=x+1$, in which  case $\rho_\ell=\sup_{x\in\mb{N}}{(x+2)}/{(x+1)}=3/2<2=\sup_{x\in \mb{R}_{>0}}{(2x+1)}/{(x+1)}=\mu_\ell.$

%% file: parts-body/conclusions.tex
\section{Discussion and Conclusions}

{Congestion games provide a fundamental framework to study decision-making in the presence of congestion effects. Prior to our work, little was known regarding the classical problem of minimizing the \socialcost in congestion games, both in the case where we allow for centralized decision-making and for the case where \players{} can be coordinated only indirectly through the use of interventions.

In these settings, we provide conclusive answers on what approximations can be achieved by efficient (i.e., polynomial time) algorithms, and what are the inherent computational limitations. Interestingly, our technical results show that \emph{no} performance degradation arises when moving from centralized decision-making to the use of interventions. On the contrary, judiciously designed taxation mechanisms can be efficiently computed and achieve the same performance of the best centralized polynomial time algorithm. We achieve this result by providing a tight computational lower bound for the problem of minimizing the \socialcost, and by designing suitable taxation mechanisms with matching performance. We thus obtain polynomial time algorithms based on taxes matching the hardness factor.}

There remain many opportunities for further work on interventions in congestion games. One important research direction is to shed further light on the interplay between (more general) interventions and the achievable performances. It is interesting to understand whether other approaches based on, e.g., information provisioning, are equally powerful. %
In other words, does the positive result obtained here for taxation mechanisms hold for other classes of interventions?

There are also a number of open questions arising from our work and  possible refinements thereof. 
{A first direction is that of considering the {\it weighted} version of congestion games, where \players{} increase the congestion of a resource depending on their weight. While the hardness result we obtained in \cref{thm:hardness} holds for this more general class of problems, the corresponding hardness factor can be strengthened (i.e., made larger) by considering a partitioning system where all \players{} share a common weight $w$, not necessarily equal to one. On the other hand, the question of designing taxation mechanisms with optimal performance is significantly more challenging, and we leave this open as a direction for future work.}
A second direction} is that of considering the variant of \emph{network congestion games}, whereby each strategy set $\mc{A}_i$ is implicitly given as the set of paths connecting an origin/destination node in an underlying graph. 
On its own, this more succinct representation of the strategy space would only increase the computational complexity, but on the other hand the graph also imposes more structure.
The results in \cref{thm:opttaxes} (design of optimal taxes) extend to network congestion games by replacing the constraint set in the convex program in \eqref{eq:cvxprogram} with the set of feasible flows on the graph. Similarly, \cref{cor:polyalgo} (polynomial time algorithms) also extends provided that one utilizes no-regret algorithms, such as Follow the Perturbed Leader, that do not require explicit description of all possible paths (which might be exponential in the size of the graph). On the contrary, our reduction in the proof of \cref{thm:hardness} (inapproximability of minimum cost) induces a general congestion game. The existence of a reduction to network congestion games remains open.

%% file: parts-body/ack-biblio.tex
\bibliographystyle{ACM-Reference-Format}
\bibliography{references.bib}

%% file: parts-appendix/appendix-EC.tex
\begin{appendices}
\section{Appendix}

{As in the main body of the manuscript, throughout the appendix we extend the domain of definition of all resource costs from $\mb{N}$ to $\mb{N}_0$ by setting their value to zero.  This is without loss of generality, and only used to ease the notation, as already clarified on page \cpageref{foot:simplify_b(0)} and \cref{foot:simplify_b(0)}.} %

\subsection{Additional material for Section~\ref{sec:hardness}}
\subsubsection{Binomial expected value converges to Poisson expected value}
\label{app:binomial-poisson}
\begin{lemma}
\label{lem:convergence}
Let $\ell:\mb{N}\rightarrow\mb{R}_{\ge0}$ satisfy {Assumption~\ref{ass:assumption}}, and $\rho_\ell<\infty$. Then $\forall k\in\mb{N}$
\[
\lim_{h\to\infty}\mb{E}_{X\sim \emph{Bin}(h,\el/h)} [X\ell(X)] = 
\mb{E}_{X\sim \emph{Poi}(\el)} [X\ell(X)].
\]
\end{lemma}

\begin{proof}{Proof.}
{We} begin noting that the limiting operation is delicate since we do not want to assume boundedness of $x\ell(x)$ for $x\rightarrow \infty$, 
as this would disqualify interesting cases such as that of polynomials. %
Since $\rho_\ell<\infty$ by assumption, then $\mb{E}_{X\sim \text{Poi}(\el)} [X\ell(X)]<\infty$ for any $k\in\mb{N}$. Fix $k$, and let $f_h(i)$ and $f(i)$ denote the probability mass function corresponding to $\text{Bin}(h,k/h)$ and $\text{Poi}(k)$, and recall that $f_h(i)\rightarrow f(i)$ for all $i$. Define $c(x)=x\ell(x)$ for $x\in\mb{N}$ and $c(0)=0$. \mbox{The result follows from}
\[
\begin{split}
\mb{E}_{X\sim \text{Poi}(k)}[c(X)]
&\!=\!\!
\sum_{i=0}^\infty \lim_{h\rightarrow\infty }f_h(i)c(i)
\!
\le 
\!\!
\lim_{h\rightarrow\infty}\sum_{i=0}^\infty f_h(i)c(i) 
\!
=
\!
\lim_{h\rightarrow\infty}\!
\mb{E}_{X\sim  \text{Bin}(h,k/h)}[c(X)]
\!\le
\!
\mb{E}_{X\sim \text{Poi}(k)}[c(X)],
\end{split}
 \]
 where the first equality holds by definition of expected value and by replacing $f(i)=\lim_{h\rightarrow\infty}f_h(i)$. The following inequality holds by Fatou's lemma and existence of the limit (which hold as this is a series with non-negative terms). As a result, we can interchange the limiting operation with the infinite sum. The next equality is due to the definition of expected value. The last inequality is a consequence of the fact that $c$ is a convex function and the Poisson distribution $\text{Poi}(k)$ dominates the binomial distribution $\text{Bin}(h,k/h)$ in the sense of the convex ordering \citep{shaked2007stochastic}, ensuring that $\mb{E}_{X\sim  \text{Bin}(h,k/h)}[c(X)]\le \mb{E}_{X\sim \text{Poi}(k)}[c(X)]$ for all $h$. {One way to see this is to utilize the fact that the ratio between the probability mass functions $f_h(i)/f(i)$ is unimodal as shown in \cite[Sec 2.7]{klenke2010stochastic}, and that $f_h$ and $f$ are not ordered by the usual stochastic order, thus concluding thanks to \cite[Thm. 3.A.53]{shaked2007stochastic}. Alternatively, one can compare the two expectations directly.} 
\hfill\qed
\end{proof}

\subsubsection{Proof of \cref{prop:partition-exists}}
\label{subsub:app-proof-partition-exists}
\begin{proof}{Proof.}
Existence of a partitioning system is proved through a probabilistic approach similarly to that in \cite{Feige98}. The idea is to construct each $\mc{P}_i$ independently from a uniform distribution. Formally, each element in $[n]$ is assigned to $\el$ of the $P_{j,i}$ uniformly at random. This ensures that, by construction, each element in $[n]$ appears in exactly $\el$ different sets of $\mc{P}_j$. Thus the first property of the partitioning system holds trivially. 
Define $c(x)=x\ell(x)$ for $x\in\mb{N}$ and $c(0)=0$. In order to prove the second property, we fix $B\subseteq [\beta]$ with $|B|=h$ and correspondingly consider $Q=\{P_{j,i(j)}, j\in B\}$. We intend to show that with high probability $\sum_{r\in Q} c(|Q|_r)\ge \left(\mb{E}_{X\sim \text{Bin}(h,\el /h)}[c(X)]-\eta\right)n$ holds. This would imply that, among all possible ways of constructing $\{\mc{P}_i\}_{i=1}^r$ there exists at least one that satisfies the property. To prove that $\sum_{r\in Q} c(|Q|_r)\ge \left(\mb{E}_{X\sim \text{Bin}(h,\el /h)}[c(X)]-\eta\right)n$ holds with high probability, we compute the expected cost that arises from the probabilistic choice outlined above for $\{\mc{P}_i\}_{i=1}^r$. In particular 
\[
\mb{E}[c(|Q|_r)]=\mb{E}_{X\sim \text{Bin}(h,k/h)}[c(X_r)],
\]
since the number of sets in which each resource appears is given by the random variable $X_r=\sum_{j\in B} \mathbbm{1}_{r\in P_{j,i(j)}}\sim\text{Bin}(h,\el/h)$, owing to the fact that $\mathbbm{1}_{r\in P_{j,i(j)}}\sim\text{Ber}(k/h)$ are independent (here $\mathbbm{1}$ denotes the indicator function).
We then use Chernoff-Hoeffding bound on the total cost $\sum_{r\in[n]}c(|Q|_r)$, where each term is bounded by $0\le c(|Q|_r) \le c(h)$ owing to the non-decreasingness and non-negativity of $c$.%
\footnote{Observe that the random variables $\{c(|Q|_r)\}_{r\in[n]}=\{c(X_r)\}_{r\in[n]}$ are negatively associated, which is enough to conclude thanks to \cite{dubhashi1996balls} and $\eta\in(0,1)$.  Since $c(X_r)=c(\sum_{j\in B}X_{r,j})$, where  $X_{r,j}=\mathbbm{1}_{r\in P_{j,i(j)}}$, is non-decreasing in $\{X_{r,j}\}_{j\in[\ralph]}$, negative association of $\{c(X_r)\}_{r\in[n]}$ can be shown by proving negative association of $\{X_{r,j}\}_{r\in [n], j\in[\ralph]}$ \cite[Property 6]{joag1983negative}.
For fixed $j\in[\ralph]$, the variables $\{X_{r,j}\}_{r\in [n]}$ are negatively associated as they are a permutation distribution of $(0,\dots,0,1,\dots,1)$ with $n-kn/h$ zeros and $kn/h$ ones \cite[Def. 2.10 and Thm. 2.11]{joag1983negative}. Owing to this, and thanks to \cite[Property 7]{joag1983negative} negative association of $\{X_{r,j}\}_{r\in [n], j\in[\ralph]}$ follows from the above, and from the fact that $\{X_{r,j}\}_{r\in [n]}$ are mutually independent.}
 Thus, with probability smaller or equal to $2e^{{-2n\eta^2}/{(c(h))^2}}$, it is  
$
|\sum_{r\in[n]}c(|Q|_r)-\mb{E}_{X\sim \text{Bin}(h,k/h)}[c(X)]|\ge \eta n.
$
Since there are $\binom{\ralph}{h}\cdot h^h\le (1+h)^\ralph$ possible choices for $B$ and $Q$, a union bound guarantees that with probability higher than   $1-2(1+h)^\beta\cdot e^{{-2n\eta^2}/{(c(h))^2}}$, the cost of all $B,Q$ satisfies $|\sum_{r\in[n]}c(|Q|_r)-\mb{E}_{X\sim \text{Bin}(h,k/h)}[c(X)]|<\eta n$. With the specific choice of $n$ as in the statement, we are guaranteed this property with a probability of at least $4/5$. This shows that a partitioning system always exists. One such object can be computed by simple enumeration over all possible choices, which are only a function of $h$, $n$ and $\ralph$.
\hfill\qed
\end{proof}

\subsubsection{Hardness factor for polynomial resource cost}
\label{app:factor-poly}
{
\begin{lemma}
\label{lem:factor_polynomials}	
Let $d\in\mathbb{N}_0$, and let $\ell:\mb{N}\rightarrow \mb{R}_{\ge0}$ be $\ell(x)=\sum_{j=0}^d \alpha_j x^j$ where $\alpha_j\ge0$ for all $j$. Let $m:\mb{N}\rightarrow \mb{R}_{\ge0}$ be $m(x)=x^d$. Define $\rho_\ell$, $\rho_m$ as in \eqref{eq:approxfact}. Then $\rho_\ell\le \rho_m$.
\end{lemma}
\begin{proof}{Proof.}
Using the definition of $\rho_\ell$, $\rho_m$, we need to show that
\[
\rho_\ell=\sup_{x\in\mb{N}}\frac{\mb{E}_{P\sim \text{Poi}(x)}[\sum_{j=0}^d \alpha_j P^{j+1}]}{\sum_{j=0}^d \alpha_j x^{j+1}} \le \sup_{x\in\mb{N}}\frac{\mb{E}_{P\sim \text{Poi}(x)}[P^{d+1}]}{x^{d+1}}=\rho_{m}.
\]
Towards this goal, we will prove that for any $x\in\mb{N}$, it is %
$
\sum_{j=0}^d \alpha_j \mb{E}_{P\sim \text{Poi}(x)}\left[P^{j+1}\right]x^{d+1} \le \sum_{j=0}^d \alpha_j x^{j+1} \, \mb{E}_{P\sim \text{Poi}(x)}[P^{d+1}]$, from which one concludes by rearranging and taking the supremum on both sides. We will do so term by term, showing that for all $j\in\{0,\dots,d\}$ it is
\[
\mb{E}_{P\sim \text{Poi}(x)}\left[P^{j+1}\right]x^{d-j} \le \mb{E}_{P\sim \text{Poi}(x)}[P^{d+1}].
\] Towards this goal, we recall that the $(j+1)$'st moment of the Poisson distribution $\text{Poi}(x)$ equals $\sum_{i=0}^{j+1}x^i {j+1 \brace i}$, where ${j+1 \brace i}$ is a Stirling number of the second kind \cite[p. 63]{mansour2015commutation}. Hence,
\[
\mb{E}_{P\sim \text{Poi}(x)}\left[P^{j+1}\right]x^{d-j} = \sum_{i=0}^{j+1}x^{i+d-j} {j+1 \brace i} = \sum_{i=d-j}^{d+1}x^{i} {j+1 \brace i-(d-j)} \le \sum_{i=0}^{d+1}x^{i} {d+1 \brace i} = \mb{E}_{P\sim \text{Poi}(x)}\left[P^{d+1}\right],
\]
where the key inequality ${j+1 \brace i-(d-j)} \le  {d+1 \brace i} $ holds by the recurrence relation defining the Stirling numbers, and the fact that $j\le d$ \cite[Thm 1.17]{mansour2015commutation}.\hfill\qed
\end{proof}
}

\subsection{Additional material for Section~\ref{sec:taxes}}

\subsubsection{Convexity of Poisson expected value}
\label{app:pois-cvx}
\begin{lemma}
\label{lem:pois-is-cvx}
Let $\ell:\mb{N}\rightarrow \mb{R}_{\mygezero}$ satisfy {Assumption~\ref{ass:assumption}} and  $\rho_\ell$ in \eqref{eq:approxfact} be finite. Define $\hfun:\mb{R}_{\ge0}\rightarrow \mb{R}_{\ge0}$ as
\[
\hfun(v)=\mb{E}_{P\sim\emph{Poi}(v)}[P\ell(P)] = \left(\sum_{i=0}^\infty ib(i)\frac{v^i}{i!}\right)e^{-v}.
\]	
Then, $\hfun$ is convex and differentiable infinitely many times in $\mb{R}_{\ge0}$.
\end{lemma}
\begin{proof}{Proof.}
We start by showing that if $\rho_\ell<\infty$, then $\hfun(v)$ is well defined, i.e., the series converges to a finite value for any choice of $v\in\mb{R}_{\ge0}$ (as standard, $\mb{R}$ does not include infinity). For $v=0$, this is immediate. For $v>0$, define $c(v)=v\ell(v)$ for $v\in\mb{N}$ with $c(0)=0$, and observe that $\rho_\ell<\infty$ implies $\hfun(v)/c(v)<\infty$ for all fixed $v\in\mb{N}$ thanks to the definition of $\rho_\ell$, so that $\hfun(v)<\infty$ since $0<c(v)<\infty$ for finite $v\in\mb{N}$. Therefore also $\sum_{i=0}^\infty c(i)\frac{v^i}{i!}<\infty$ for all $v\in\mb{N}$. Since $\sum_{i=0}^\infty c(i)\frac{v^i}{i!}$ is increasing in $v\ge0$, boundedness over the naturals, immediately implies boundedness of the same expression over the non-negative reals. Exploiting the fact that also $e^{-v}$ is bounded, we have shown that $\hfun(v)$ is well defined, and converges to a finite value for any choice of $v\in\mb{R}_{\ge0}$. As a result $\hfun$ is differentiable infinitely many times in its domain, since it is the product of a convergent power series, and of $e^{-v}$. We can therefore prove convexity of $p(v)$ by computing its second order derivative and verifying that it is non-negative. The first derivative reads as 
\[
\begin{split}
\hfun'(v)
=&-e^{-v}\left(\sum_{i=0}^\infty c(i)\frac{v^i}{i!}\right)+e^{-v}\left(\sum_{i=0}^\infty c(i)\frac{i v^{i-1}}{i!}\right)\\
=&-e^{-v}\left(\sum_{i=0}^\infty c(i)\frac{v^i}{i!}\right)+e^{-v}
\left(\sum_{i=0}^\infty c(i+1)\frac{ v^{i}}{i!}\right)\\
=&%
e^{-v}\sum_{i=0}^\infty \frac{v^i}{i!}\Delta c(i),
\end{split}
\]
where we defined $\Delta c(i)=c(i+1)-c(i)$. Following an identical approach the second derivative is 
\[
\hfun''(v)=e^{-v}\sum_{i=0}^\infty \frac{v^i}{i!}[\Delta c(i+1)-\Delta c(i)]. 
\]
The summand corresponding to $i=0$ reads as $c(2)-2c(1)=2(\ell(2)-\ell(1))>0$ since $\ell$ is non-decreasing. The summands corresponding to $i\ge 1$ are non-negative as semi-convexity of $\ell$ implies convexity of $c$. Hence, $\hfun''(v)\ge0$ for all $v\in\mb{R}_{\ge0}$ as desired.
\hfill\qed
\end{proof}

\subsubsection{Proof of Lemma \ref{lem:propertiesTM}}
\label{app:twoproperties}
\hfill\\
{\Cref{lem:propertiesTM} consists of three parts. We
first prove part (a), then (b), and finally (c).}

\medskip
\begin{proof}{$\triangleright$ Proof of Part (a).}
We will show that $f_r(x,v)\ge \ell_r(x)$ for any $x\in\mb{N},v\in\mb{R}_{\ge0}, r\in\mathcal{R}$, as this suffices to conclude since $\tau_r(x,v_r)=f_r(x,v)- \ell_r(x)$. We do so separately for each resource, and thus drop the index $r$ in the following. The case of $v=0$ follows readily, since $f(x,0)=\ell(x)$. Similarly, for $x=0$, it is $f(0,v)=0\ge \ell(0)=0$. In the remaining cases, we are left to prove 
\[
f(x,v)=\frac{(x-1)!}{v^x} \sum_{i=0}^{x-1} \frac{\hfun(v)-i\ell(i)}{i!}v^i\ge \ell(x)\qquad x\in\mb{N}, v\in\mb{R}_{>0}.
\]	
Define $c(x)=xb(x)$ for $x\in\mb{N}$ with $c(0)=0$. The latter inequality holds if  
$
\hfun(v)\ge ({\sum_{i=0}^x c(i)\frac{v^i}{i!}})/({\sum_{j=0}^{x-1}
\frac{v^j}{j!}})
$
for all $x\in\mb{N}$, $v\in \mb{R}_{>0}$. We complete the proof showing that the right hand side in the previous expression is non-decreasing in $x$, so that the desired property holds for any $x\in\mb{N}$ if it holds for $x$ arbitrarily large, i.e., if
\[
\hfun(v)\ge \lim_{x\rightarrow \infty} \frac{\sum_{i=0}^x c(i)\frac{v^i}{i!}}{\sum_{j=0}^{x-1}
\frac{v^j}{j!}
},\qquad \forall v\in \mb{R}_{>0},
\]
which is indeed satisfied (with equality), as it reduces to $\hfun(v)\ge (\hfun(v)e^v)/e^v$, owing to the definition of $\hfun(v)$ and $\sum_{j=0}^\infty \frac{v^j}{j!}=e^v$. To conclude, we thus need to prove that for all $x\in\mb{N}$, $v\in\mb{R}_{>0}$, it is 
\[
\begin{split}
\frac{\sum_{i=0}^{x+1} c(i)\frac{v^i}{i!}}{\sum_{j=0}^{x}\frac{v^j}{j!}} 
\ge
\frac{\sum_{i=0}^{x} c(i)\frac{v^i}{i!}}
{\sum_{j=0}^{x-1}\frac{v^j}{j!}}
&\iff 
\sum_{i=0}^{x+1}\sum_{j=0}^{x-1} c(i)\frac{v^{i+j}}{i!j!}
\ge
\sum_{i=0}^{x}\sum_{j=0}^{x} c(i)\frac{v^{i+j}}{i!j!} \\
&\iff 
\sum_{j=0}^{x-1} c(x+1) \frac{v^{x+1+j}}{(x+1)!j!} 
\ge
\sum_{i=0}^x c(i) \frac{v^{i+x}}{i!x!}\\
& \iff 
\sum_{j=1}^{x} c(x+1) \frac{v^{x+j}}{(x+1)!(j-1)!} 
\ge
\sum_{i=1}^x c(i) \frac{v^{i+x}}{i!x!}\\
&\iff
\sum_{j=1}^x \left[\frac{c(x+1)}{(x+1)!(j-1)!} - \frac{c(j)}{x!j!}\right]v^{j+x} \ge0\\
&\iff \sum_{j=1}^x [\ell(x+1)-\ell(j)]v^{j+x}\ge0%
\end{split}
\]
which follows from the above chain of implications, and the fact that $\ell(x)$ is non-decreasing.
\hfill\qed
\end{proof}

\medskip

\begin{proof}{$\triangleright$ Proof of Part (a).}
To ease the notation, we drop the subscript $r$, i.e., we show that $f(x+1,v)\ge f(x,v)$ for all $x\in\mb{N}_0$ and $v\in\mb{R}_{\ge 0}$. When $v=0$, then $f(x,0)=\ell(x)$, which is non-decreasing as $\ell(x)$ is so. Thus, in the following we restrict to the case of $v> 0$. When, in addition, $x=0$ the inequality reduces to $f(1,v)\ge f(0,v) \iff \hfun(v)/v\ge0$, which holds as $\hfun(v)>0$ for $v>0$. We are thus left to consider the case of $x\in\mb{N}$ and $v\in\mb{R}_{> 0}$.
Substituting the expression for $f(x+1,v)$ and $f(x,v)$ in the latter inequality and isolating $\hfun(v)$ results in 
\be
\label{eq:nondecr-ineq}
\hfun(v)\left(x \sum_{i=0}^x \frac{v^i}{i!} - v\sum_{i=0}^{x-1}\frac{v^i}{i!}\right)\ge 
\left(x \sum_{i=0}^x \frac{c(i)}{i!}v^i - v\sum_{i=0}^{x-1}\frac{c(i)}{i!}v^i\right),
\ee
where we defined $c(i)=i\ell(i)$ for $i\in\mb{N}$ and $c(0)=0$.
The term in brackets on the left hand side can be equivalently written as
\[
\sum_{i=0}^x x\frac{v^i}{i!} - \sum_{i=0}^{x-1}\frac{v^{i+1}}{i!}
=
\sum_{i=0}^x \frac{x-i}{i!}v^i.  
\]
Similarly, the term in the right hand side brackets, is equivalent to
\[
\sum_{i=1}^x x\frac{c(i)}{i!}v^i - \sum_{i=0}^{x-1}\frac{c(i)}{i!}v^{i+1} 
=
\sum_{i=1}^x \frac{x c(i)-i c(i-1)}{i!}v^i.
\]
Thus, inequality \eqref{eq:nondecr-ineq} reduces to 
\be
\hfun(v)\ge 
\frac
{\sum_{i=1}^x \frac{x c(i)-i c(i-1)}{i!}v^i}
{\sum_{i=0}^x \frac{x-i}{i!}v^i},
\label{eq:h_gek}
\ee
whereby we used the fact that the denominator is positive since $x\ge1$ and $x\ge i$.
Thus, we require \eqref{eq:h_gek} to hold for all $x\in\mb{N}$ and $v\in\mb{R}_{> 0}$.
\Cref{lem:non-decr}, ensures that the right hand side of the previous inequality is non-decreasing in $x\in\mb{N}$, for each fixed $v\in\mb{R}_{> 0}$. Therefore, \eqref{eq:h_gek} holds, if it holds when $x$ is arbitrarily large, that is if
\be
\hfun(v)\ge \lim_{x\to\infty} \frac{\sum_{i=0}^x\frac{c(i)}{i!}v^i}{\sum_{i=0}^x\frac{v^i}{i!}}.
\label{eq:final_limit}
\ee
Notice, though, that the right hand side in the last expression is precisely equal to $\frac{\hfun(v)e^v}{e^v}=\hfun(v)$, thanks to the definition of $\hfun(v)$ in \eqref{eq:binExp} and to the fact that $\sum_{i=0}^\infty\frac{v^i}{i!}=e^v$. 
Therefore \eqref{eq:final_limit} holds (with equality), which completes the proof of part (a).
\hfill\qed
   \end{proof}

\medskip
\begin{proof}{Proof of Part (c).}
Using the definition of $f_r$ in \eqref{eq:f-designed}, 
observe that
\begin{align*}
   v f_jr(x+1,v) - x f_r(x,v)
   &= v\frac{x!}{v^{x+1}}\sum_{i=0}^{x} \frac{\hfun_r(v)-i\ell_r(i)}{i!}{v^i}
     - x\frac{(x-1)!}{v^x}\sum_{i=0}^{x-1} \frac{\hfun_r(v)-i\ell_r(i)}{i!}{v^i}\\
   &= \frac{x!}{v^{x}} \frac{\hfun_r(v)-x\ell_r(x)}{x!}{v^x}\\
   &= \hfun_r(v)-x\ell_r(x),
\end{align*}
or equivalently $x \ell_r(x)-x f_r(x,v)+v f_r(x+1,v)=\hfun_r(v)$ as needed.
\hfill\qed
\end{proof}

\subsubsection{Technical Lemma used to prove Lemma \ref{lem:propertiesTM}}
\begin{lemma}
\label{lem:non-decr}
Let $c:\mb{N}\rightarrow\mb{R}_{\ge0}$ be convex. Then, the function $g:\mb{N}\times\mb{R}_{>0}\rightarrow \mb{R}_{\ge0}$, 
\[
g(x,v)=\frac
{\sum_{i=1}^x \frac{x c(i)-i c(i-1)}{i!}v^i}
{\sum_{i=0}^x \frac{x-i}{i!}v^i},
\]
where we set $c(0)=0$, is non-decreasing for all $x\in\mb{N}$, for any fixed $v\in\mb{R}_{>0}$.
\end{lemma}
\begin{proof}{Proof.}
Proving the claim amounts so showing that for all $v\in\mb{R}_{>0}, x\in\mb{N}$ it is
\be
\frac
{\sum_{i=1}^{x+1} a(x+1,i)\frac{v^i}{i!}}
{\sum_{j=0}^{x+1} b(x+1,j)\frac{v^j}{j!}}
\ge
\frac{\sum_{i=1}^x a(x,i)\frac{v^i}{i!}}
{\sum_{j=0}^x b(x,j)\frac{v^j}{j!}}
,
\label{eq:inequalitykkp1}
\ee
where we let $a(x,i)= xc(i)-ic(i-1)$ and $b(x,j)=x-j$ to ease the notation. Since the denominators on the left and right hand side are positive, and since $b(x+1,x+1)=0$,  \eqref{eq:inequalitykkp1} is equivalent to
\[
\sum_{i=1}^x\sum_{j=0}^x a(x+1,i)b(x,j)\frac{v^{i+j}}{i!j!} 
+
\sum_{j=0}^x b(x,j)a(x+1,x+1)\frac{v^{j+x+1}}{j!(x+1)!}
-
\sum_{i=1}^x\sum_{j=0}^xa(x,i)b(x+1,j)\frac{v^{i+j}}{i!j!}
\ge 0,
\]	
which we rewrite 
\be
\underbrace{\sum_{i=1}^x\sum_{j=0}^x \left(a(x+1,i)b(x,j)-
a(x,i)b(x+1,j)\right)\frac{v^{i+j}}{i!j!}}_{{\footnotesize\circled{1}}}
+
\underbrace{\sum_{j=0}^x b(x,j)a(x+1,x+1)\frac{v^{j+x+1}}{j!(x+1)!}}_{{\footnotesize\circled{2}}}
\ge 0.
\label{eq:thelastone}
\ee
We now turn our attention to each of the two terms appearing in the previous inequality. In particular, we will show that collecting all the contributions corresponding to the same power of $v$ significantly simplifies the expressions, and allows us to conclude. 

We begin with the second term, and substitute the definitions of $a$, $b$ in $b(x,j)a(x+1,x+1)$ so that 
\be
{\footnotesize\circled{2}}=\sum_{j=0}^x \frac{x-j}{j!x!}(c(x+1)-c(x))v^{j+x+1}.
\label{eq:circ2}
\ee

We now focus on the first term, and observe that 
\[
\begin{split}
a(x+1,i)b(x,j)-a(x,i)b(x+1,j)
&\!=\!
((x+1)c(i)-ic(i-1))(x-j)\!-\!(xc(i)-ic(i-1))(x+1-j)\\
&\!=\!-jc(i)+ic(i-1),
\end{split}
\]
where we made use of the definitions of $a$ and $b$. %
Hence, we can utilize indices $i$ and $\p=i+j$ in place of $i$, $j$ to rewrite the first term appearing in \eqref{eq:thelastone} as
\be
\begin{split}
\label{eq:circ1}
{\footnotesize\circled{1}}=&\sum_{i=1}^x\sum_{j=0}^x (ic(i-1)-jc(i))\frac{v^{i+j}}{i!j!}\\
=&\sum_{\p=1}^{2x}\sum_{\substack{i\in [x]\\\text{s.t.~} \p-x\le i\le \p}}
\frac{ic(i-1)-(\p-i)c(i)}{i!(\p-i)!}v^{\p}\\
=&
\sum_{\p=x+1}^{2x}\sum_{i=\p-x}^x
\frac{ic(i-1)-(\p-i)c(i)}{i!(\p-i)!}v^{\p}
+
\sum_{\p=1}^{x}\sum_{i=1}^\p
\frac{ic(i-1)-(\p-i)c(i)}{i!(\p-i)!}v^{\p} \\
=&\sum_{\p=x+1}^{2x}\sum_{i=\p-x}^x
\frac{ic(i-1)-(\p-i)c(i)}{i!(\p-i)!}v^{\p}\\
=&\sum_{\p=x+1}^{2x}\frac{c(\p-x-1)-c(x)}{x!(\p-x-1)!}v^\p\\
=&\sum_{j=0}^{x-1}\frac{1}{j!x!}(c(j)-c(x))v^{j+x+1}
\end{split}
\ee
where the second line is obtained using the fact that $\p$ runs from $1$ to $2x$, $i\in[x]$ and $j=\p-i$ belongs to $0\le \p-i\le x$. The third line follows upon distinguishing the case of $x+1\le \p\le 2x$ and $\p\in [x]$. The fourth line is due to the fact that the second summand in the third line vanishes, since 
\[
\sum_{i=1}^\p
\frac{ic(i-1)-(\p-i)c(i)}{i!(\p-i)!}v^{\p} =
\sum_{i=2}^\p\frac{c(i-1)}{(i-1)!(\p-i)!}-\sum_{i=1}^{\p-1} \frac{c(i)}{i!(\p-i-1)!}
=0.
\]
The fifth line follows from 
\[
\begin{split}
\sum_{i=\p-x}^x&
\frac{ic(i-1)-(\p-i)c(i)}{i!(\p-i)!}
=
\sum_{i=\p-x}^x
\frac{c(i-1)}{(i-1)!(\p-i)!}
-
\sum_{i=\p-x}^x
\frac{c(i)}{i!(\p-i-1)!}\\
&=\frac{c(\p-x-1)}{(\p-x-1)!x!} + 
\sum_{i=\p-x+1}^x
\frac{c(i-1)}{(i-1)!(\p-i)!}
-
\sum_{i=\p-x}^{x-1}
\frac{c(i)}{i!(\p-i-1)!}
- \frac{c(x)}{x!(\p-x-i)!}\\
&=\frac{c(\p-x-1)-c(x)}{x!(\p-x-1)!}
\end{split}
\]
The final line is derived reverting to the original indices $i$ and $j$.

Thus, in light of \eqref{eq:circ1} and \eqref{eq:circ2}, the inequality \eqref{eq:inequalitykkp1} we need to show reduces to
\[
\sum_{j=0}^{x-1}[c(j)-c(x)+ (x-j)(c(x+1)-c(x))]\frac{v^{j+x+1}}{j!x!}\ge0\quad\forall v\in\mb{R}_{>0},~x\in\mb{N}.
\] 
The summand corresponding to $j=0$ reads as $xc(x+1)-(x+1)c(x)=x(x+1)(b(x+1)-b(x))$, and it is non-negative since $b$ is non-decreasing. All other summands are non-negative thanks to the convexity of $c$, that guarantees $c(x+1)-c(x) \ge (c(x)-c(j))/(x-j)$ as $1\le j<x$, so that
$c(j)-c(x)+ (x-j)(c(x+1)-c(x))\ge0$. One concludes using these observations and $v>0$.
\hfill\qed
\end{proof}

\end{appendices}